\documentclass[11pt]{article}

\PassOptionsToPackage{unicode}{hyperref}
\PassOptionsToPackage{hyphens}{url}
\PassOptionsToPackage{dvipsnames,svgnames,x11names}{xcolor}

\usepackage{setspace}
\renewenvironment{abstract}{
  \begin{center}
  \textbf{\large Abstract}
  \end{center}
}{\par}

\usepackage{iftex}
\ifPDFTeX
  \usepackage[T1]{fontenc}
  \usepackage[utf8]{inputenc}
  \usepackage{textcomp} 
\else 
  \usepackage{unicode-math}
  \defaultfontfeatures{Scale=MatchLowercase}
  \defaultfontfeatures[\rmfamily]{Ligatures=TeX,Scale=1}
\fi

\usepackage{lmodern}

\usepackage{amsmath,amssymb,amsthm,dsfont,bbm}

\theoremstyle{plain}
\newtheorem{theorem}{Theorem}

\newtheorem{assumption}{Assumption}

\theoremstyle{plain}
\newtheorem{lemma}{Lemma}

\usepackage[authoryear]{natbib}

\IfFileExists{upquote.sty}{\usepackage{upquote}}{}
\IfFileExists{microtype.sty}{%
  \usepackage[]{microtype}
  \UseMicrotypeSet[protrusion]{basicmath} 
}{}

\makeatletter
\@ifundefined{KOMAClassName}{
  \IfFileExists{parskip.sty}{%
  }{
    \setlength{\parindent}{0pt}
    \setlength{\parskip}{6pt plus 2pt minus 1pt}}
}{
  \KOMAoptions{parskip=half}}
\makeatother

\usepackage[margin=1in]{geometry}
\usepackage{longtable,booktabs,array}
\usepackage{multirow}
\usepackage{float}
\usepackage{calc} 

\usepackage{soul}

\usepackage{etoolbox}
\makeatletter
\patchcmd\longtable{\par}{\if@noskipsec\mbox{}\fi\par}{}{}
\makeatother

\IfFileExists{footnotehyper.sty}{%
  \usepackage{footnotehyper}%
}{%
  \usepackage{footnote}%
}
\makesavenoteenv{longtable}

\usepackage{graphicx}
\makeatletter
\def\maxwidth{\ifdim\Gin@nat@width>\linewidth\linewidth\else\Gin@nat@width\fi}
\def\maxheight{\ifdim\Gin@nat@height>\textheight\textheight\else\Gin@nat@height\fi}
\makeatother

\setkeys{Gin}{width=\maxwidth,height=\maxheight,keepaspectratio}

\makeatletter
\def\fps@figure{htbp}
\makeatother

\newlength{\cslhangindent}
\newlength{\csllabelwidth}
\newlength{\cslentryspacingunit} 
 {%
  \setlength{\parindent}{0pt}
  \ifodd #1
    \let\oldpar\par
    \def\par{\hangindent=\cslhangindent\oldpar}
  \fi
  \setlength{\parskip}{#2\cslentryspacingunit}
 }%
 {}

\ifLuaTeX
  \usepackage{selnolig}  
\fi

\IfFileExists{bookmark.sty}{\usepackage{bookmark}}{\usepackage{hyperref}}
\IfFileExists{xurl.sty}{\usepackage{xurl}}{} 
\hypersetup{
  pdftitle={Extreme Quantile Treatment effect},
  pdfauthor={Mengran Li},
  colorlinks=true,
  linkcolor={blue},
  filecolor={Maroon},
  citecolor={Blue},
  urlcolor={blue},
  pdfcreator={LaTeX via pandoc}
}

\usepackage{xcolor}

\definecolor{green}{RGB}{34,139,34}

\definecolor{ml}{RGB}{0,128,0}

\usepackage{orcidlink}

\newcommand{\1}{\mathbf{1}}
\newcommand{\E}{\mathbb{E}}
\newcommand{\R}{\mathbb{R}}

\begin{document}
\title{Tail-Calibrated Estimation of Extreme Quantile Treatment Effects}

\author{Mengran Li* and Daniela Castro-Camilo}
\date{\emph{School of Mathematics and Statistics, University of Glasgow, UK.}}

\author{
Mengran Li\orcidlink{0009-0001-8273-9321}
\and
Daniela Castro-Camilo\orcidlink{0000-0002-7536-4613}
}

\footnotetext[1]{
\baselineskip=10pt m.li.3@research.gla.ac.uk}

\maketitle

\begin{abstract}
Extreme quantile treatment effects (eQTEs) measure the causal impact of a treatment on the tails of an outcome distribution and are central for studying rare, high-impact events.
Standard QTE methods often fail in extreme regimes due to data sparsity, while existing eQTE methods rely on restrictive tail assumptions or on interior-quantile theory.
We propose the Tail-Calibrated Inverse Estimating Equation (TIEE) framework, which combines information across quantile levels and anchors the tail using extreme value models within a unified estimating equation approach.
We establish asymptotic properties of the resulting estimator and evaluate its performance through simulation under different tail behaviours and model misspecifications. An application to extreme precipitation in the Austrian Alps illustrates how TIEE enables observational causal attribution for very rare events under anthropogenic warming.
More broadly, the proposed framework establishes a new foundation for causal inference on rare, high-impact outcomes, with relevance across environmental risk, economics, and public health.
\end{abstract}
\vspace{1cm}


\section{Introduction}\label{introduction}

Estimating causal effects at extreme quantiles is central in applications where rare, high-impact outcomes drive risk and decision-making. 
Examples arise in climate attribution, financial losses, environmental hazards and health extremes, where the interest lies in changes in the severity of the most extreme events rather than in average effects. 
A natural framework for studying distributional causal effects is that of quantile treatment effects (QTEs), widely explored in classical statistics and econometrics \citep{doksum1974empirical, lehmann1975nonparametrics}.
Important examples include that of \cite{firpo2007}, who developed semiparametric estimation of QTEs under unconfoundedness, and \citep{ma2022}, who extended Firpo's work to more complex data structures and missingness.
\cite{zhang2018} formalised the conditions under which these types of estimators retain their desired asymptotic properties and concluded that asymptotic normality can be recovered for intermediate quantiles, but inference becomes ill-posed as the target quantile moves further into the tail.
Recently, \cite{cheng2024} proposed a method for causal quantile estimation via a general inverse estimating equation framework.
Their approach provides a unifying perspective on quantile identification, but their theory primarily applies to interior quantiles.
Formal causal analysis of tail quantiles remains, therefore, methodologically challenging.

Recent approaches to causal inference at extreme quantile levels have turned to extreme value theory (EVT) to facilitate estimation and inference for extreme quantiles beyond the range of observed values. 
For heavy-tailed distributions, the tail quantile function follows a power-law representation governed by the extreme value index (EVI), which is typically estimated using classical procedures such as the Hill estimator \citep{hill1975simple}. 
Within this framework,  \cite{wang2012estimation} and \cite{xu2022extreme} extended
EVT-based methods to conditional and regression settings, while \cite{deuber2024estimation} were among the first to explicitly link EVT with causal inference for quantile treatment effects (QTEs) in observational studies.
Although this work constitutes a significant advancement, the approach relies on heavy-tailed (Fr\'echet-domain) assumptions and employs a two-step estimation procedure followed by extrapolation. 

Taken together, the literature reveals that causal identification and extreme value asymptotics have been developed largely in isolation, with existing methods treating confounding adjustment and tail extrapolation sequentially. Consequently, no unified approach currently delivers formally identified and statistically valid inference for causal effects at extreme quantiles.
To address this, we introduce the Tail-Calibrated Inverse Estimating Equation (TIEE) framework, a new inferential paradigm for causal analysis in the tails. TIEE integrates causal identification and extreme value modelling within a single estimating equation, aggregating information across tail quantile levels to yield a well-defined and stable target even in regimes where conventional methods fail.

Our contributions are fourfold. First, we propose a general class of tail-calibrated estimating equations that enables identification of extreme quantile treatment effects without restrictive tail assumptions. Second, we develop a computationally stable estimator that leverages EVT to inform extrapolation beyond observed extremes. Third, we establish identification, consistency, and asymptotic normality, with variance expressions accounting for both causal and tail uncertainty. Finally, through extensive simulations and an application to extreme precipitation in the Austrian Alps, we demonstrate substantial gains in efficiency and robustness over existing methods.

The remainder of the paper is organised as follows. Section~\ref{sec:background} reviews causal inference, extreme quantiles, and the IEE framework. Section~\ref{sec:TIEE} introduces TIEE and its integration with causal signals and EVT. Section~\ref{sec:property} presents theoretical results and inference procedures. Section~\ref{sec:simulation} reports simulations, and Section~\ref{sec:app} provides an application to Austrian precipitation data. Section~\ref{sec:conclusion} concludes.

\noindent\textbf{A note on notation.}
Throughout the paper, $D \in \{0,1\}$ denotes a binary treatment indicator, with $d \in \{0,1\}$ its generic value, and $Y_d$ the corresponding potential outcome (subscript omitted when clear). 
The vector $\pmb{X}$ denotes covariates, and the data consist of independent copies $W(i)=(Y(i),D(i),\pmb{X}(i))$, $i=1,\ldots,n$, of $W=(Y,D,\pmb{X})$.
We write $F_{Y,d}(\cdot)$ for the marginal distribution of $Y$ under treatment $d$.
Quantile levels are denoted by $\tau$ or $p \in (0,1)$, and treated as fixed. 
The corresponding quantile is $Q_{Y_d}(\tau)=\inf\{y:F_{Y,d}(y)\geq \tau\}$, with shorthand $\theta_{d,\tau} = Q_{Y_d}(\tau)$. 
When convenient, we also write $q(p)$ for a generic quantile function.


\section{Background}\label{sec:background}
Let  $Y_1(i) - Y_0(i)$ denote the causal effect of the treatment for observation $i$.
To identify and estimate treatment effects from observational data, we require the following standard assumptions.
\begin{assumption}[Identifiability conditions]\label{ass:identifiability}
We assume: (i) {Consistency}, so the observed outcome equals the corresponding potential outcome; 
(ii) {SUTVA}, ensuring no interference between units and well-defined treatment levels; 
(iii) {Unconfoundedness}, $(Y_0,Y_1)\perp D \mid \pmb{X}$, meaning that the outcomes are independent of the treatment given a set of confounders; and 
(iv) {Overlap}, $0<\mathrm P(D=1\mid \pmb{X})<1$ for all $\pmb{X}$ in the support, meaning both treatment states must be possible for every covariate pattern.
\end{assumption}

Whilst the average treatment effect (ATE), defined as $\E\{Y_1 - Y_0\}$, remains the most commonly reported causal estimand, it only characterises the location shift between treatment distributions. 
The quantile treatment effect (QTE) provides a more comprehensive characterisation of distributional differences and is defined as
\begin{equation}\label{eq:qte}
\delta(\tau) = Q_{Y_1}(\tau) - Q_{Y_0}(\tau) = \theta_{1,\tau} - \theta_{0,\tau}
\quad \tau\in[0,1].
\end{equation}

\subsection{Existing approaches to quantile treatment effect estimation}
\label{sec:QTE_summary}
Without imposing distributional assumptions, \cite{firpo2007} proposed a QTE estimator for $\tau\in (0, 1)$ obtained by reweighting the quantile loss function as
\begin{equation*}\label{eq:firpo}
\widehat{Q}_{Y_d}(\tau) = \underset{q \in \R}{\arg\min} \sum_{i=1}^n w_{d,i} \rho_{\tau}(Y_d(i) - q),
\end{equation*}
where $\rho_{\tau}(u)=u(\tau - \1(u < 0))$ is the check function and $w_{d,i}$ 
are inverse probability weights (IPWs) constructed from the propensity score, defined as the conditional probability of receiving treatment given the observed covariates. 
An alternative approach estimates QTE via conditional distributions, as in \cite{zhang2012}. 
Specifically, the conditional distribution of the potential outcome given covariates is modelled parametrically, and the marginal distribution under treatment is obtained by averaging the fitted conditional distribution over the empirical distribution of the covariates. The desired quantile is then obtained by inverting this estimated marginal distribution. In \cite{zhang2012}, this approach is implemented by assuming a Gaussian model for the outcome after a Box–Cox transformation.

\cite{ma2022} propose estimating QTEs by linking marginal quantiles to conditional quantile regression. The approach assumes a linear quantile regression model for the conditional outcome distribution given covariates. The marginal quantile is then characterised through a moment condition that connects the marginal distribution with the conditional distribution via the law of iterated expectations. In practice, the conditional expectation is approximated by integrating over conditional quantile levels obtained from the quantile regression model. The resulting estimating equation is solved after averaging over the empirical distribution of the covariates, with the integration typically restricted to a trimmed quantile range to avoid instability near the extremes. This approach produces estimators with smaller standard errors than those of \cite{firpo2007}.

\subsection{Extreme quantile treatment effects}
Whilst QTE estimators provide a comprehensive overview of treatment effect distributions, standard asymptotic properties no longer hold when $\tau$ approaches 0 or 1. 
Considering the lower tail and following \citet{zhang2018}, we distinguish three asymptotic regimes for extreme quantile levels $\tau_n\to0$, defined by how fast the effective sample size $n\tau_n$ shrinks: intermediate ($\tau_n\to 0 \text{ and } n\tau_n\to \infty$), moderately extreme ($\tau_n\to 0 \text{ and } n\tau_n\to d>0$), and extreme ($\tau_n\to 0 \text{ and } n\tau_n\to 0$).
The first regime still admits Gaussian limits for suitably normalised estimators, whereas in the latter two regimes, empirical quantiles become unstable and classical root-$n$ theory fails.
\cite{zhang2018} established that Firpo's (\citeyear{firpo2007}) estimator is asymptotically normal for intermediate quantiles and developed $b$-out-of-$n$ bootstrap procedures for moderately extreme cases. 
However, empirical quantile-based methods become unreliable when $n\tau_n\to 0$, as few or no observations fall in the tail region.
In the supplementary materials~\citep{zhang2018supplement}, the author also proposes a Pickands-type estimator for the extreme value indices (EVIs) of the potential outcome distributions, adapted to the causal inference setting in which potential outcomes are only partially observed.
The EVI quantifies the heaviness or lightness of the tail of a distribution.
Specifically, in EVT, one typically assumes an i.i.d. sample $Y(1)\ldots, Y(n)$ from distribution $G$ that is in the max-domain of attraction of a generalised extreme value distribution, meaning there exist sequences of constants $\{a_n\ge0\}$ and $\{b_n\}$ such that
\begin{equation}\label{eq:maxstab}
\lim_{n\to\infty}G^n(a_ny+b_n)=\exp\left\{-(1+\gamma y)^{-1/\gamma}\right\},
\end{equation}
where $\gamma$ is the EVI. 
Heavy-tailed distributions correspond to $\gamma>0$, light-tailed to $\gamma=0$, and bounded tails to $\gamma<0$.
In \cite{zhang2018}, if $\hat q_d(\tau)$ denote the estimator of the $\tau$-th quantile of $Y_d$, the EVI $\gamma_d$ is estimated as
\begin{equation}
\hat{\gamma}_d
=
-\sum_{r=1}^{R} w_r
\frac{
\log\!\left(
\hat q_d(m l^{r}\tau_n)-\hat q_d(l^{r}\tau_n)
\right)
-
\log\!\left(
\hat q_d(m l^{r-1}\tau_n)-\hat q_d(l^{r-1}\tau_n)
\right)
}{\log l},
\end{equation}
where $l>0$ is a spacing parameter, $m>1$, $\{w_r\}_{r=1}^{R}$ are weights satisfying $\sum_{r=1}^{R} w_r=1$, and $\tau_n$ is an intermediate quantile index.

To enable estimation beyond the range of observed data, \cite{bhuyan2023} reconstruct the counterfactual distribution via a bulk–tail model for conditional quantiles with a data-driven threshold \(\tau_u\), followed by integration over covariates and inversion. 
This requires the correct specification of the conditional quantile process and repeated tail modelling across covariates.
By contrast, \cite{deuber2024estimation} integrate extreme value theory with causal inference via causal Hill estimators for the EVI.
Specifically, they  proposed to use
\begin{equation*}
\begin{aligned}
\widehat{\gamma}_0^H&:=\frac{1}{n \tau_n} \sum_{i=1}^n\left(\log Y(i)-\log \widehat{q}_0(1-\tau_n)\right) \frac{1-D(i)}{1-\widehat{\pi}(\boldsymbol{X}_i)} \1_{Y(i)>\widehat{q}_0(1-\tau_n)},\\
\widehat{\gamma}_1^H&:=\frac{1}{n \tau_n} \sum_{i=1}^n\left(\log Y(i)-\log \widehat{q}_1(1-\tau_n)\right) \frac{D(i)}{\widehat{\pi}(\boldsymbol{X}_i)} \1_{Y(i)>\widehat{q}_1(1-\tau_n)},
\end{aligned}
\end{equation*}
where $\widehat{\pi}(\boldsymbol{X}_i)$ denotes the estimated propensity score. Extreme quantiles are then obtained by extrapolating intermediate quantiles,
\begin{equation}
\widehat{q}_d(1-p_n)=\widehat{q}_d(1-\tau_n)\left(\frac{\tau_n}{p_n}\right)^{\widehat{\gamma}_d^H},
\end{equation}
where $\tau_n \to 0$, $n\tau_n \to \infty$ and $p_n \to 0$, $np_n \to d \geq 0$. 
This extrapolation is valid only for heavy-tailed distributions ($\gamma>0$), limiting the approach to that setting.

\subsection{The inverse estimating equation framework}
\label{sec:IEE_summary}

The inverse estimating equation (IEE) framework introduced by~\citet{cheng2024} provides a general strategy for estimating quantiles of potential outcome distributions by exploiting the defining relationship between quantiles and threshold-transformed means. 
For any threshold $\theta \in \R$, 
\begin{equation}\label{eq:tau0_IEE}
\tau_d(\theta)
=
\E\!\left[\1\{Y_d \le \theta\}\right]
=
\Pr(Y_d \le \theta ), 
\end{equation}
which is the distribution function of $Y_d$ evaluated at $\theta$.
Quantiles are defined as the inverse of $\tau_d(\theta)$.
Under Assumption~\ref{ass:identifiability} and the additional requirement of point identification in \cite{cheng2024}, $\tau_d(\theta)$ is identified from observed data through a signal function $s(W,\theta,\gamma)$ satisfying $\mathbb{E}\!\left[s_d(W,\theta,\zeta)\right] = \tau_d(\theta),$ for all $\theta\in\mathbb{R},$
where $W=(Y,D,\pmb{X})$ and $\zeta$ denotes nuisance parameters. 
Defining the estimating function $g_d(W,\tau,\theta,\zeta)=s_d(W,\theta,\zeta)-\tau$, the identity $\E[g_d(W,\tau,\theta,\zeta)]=0$ holds if and only if $\tau=\tau_d(\theta)$.
Quantiles therefore follow by inversion. Let $\theta_{d,q}^\star=Q_{Y_d}(q)$ denote the $q$-quantile of $Y_d$. By definition,
\begin{equation}\label{eq:IEE_population}
\E[g_d(W,q,\theta^\star_{d,q},\zeta)]=0,
\end{equation}
so the quantile is characterised as the root of the estimating equation, which in practice is solved using its empirical analogue.
For interior quantiles, this inversion is stable and standard asymptotics apply. Near the extremes, however, data become sparse, and the density may vanish, so small perturbations in the estimating equation can cause large changes in the root, making the inversion ill-conditioned and invalidating classical asymptotics.

\section{Tail-calibrated inverse estimating equation (TIEE) framework}\label{sec:TIEE}
This section introduces our framework for estimating extreme quantile treatment effects by reformulating the inverse estimating equation as an integrated moment condition. We impose the following regularity conditions.

\begin{assumption}[Regularity conditions on potential outcomes]\label{ass:regularity}
For each $d \in \{0,1\}$: (i) $Y_d$ is continuously distributed with density $f_{Y,d}$; 
(ii) $f_{Y,d}(y)$ is monotone in the upper tail beyond some threshold $y_0$; 
(iii) $f_{Y,d}$ is bounded away from zero on any compact subset of the support; 
(iv) $F_{Y,d}$ lies in the maximum domain of attraction of an extreme value distribution with index $\xi_d$ (see~\eqref{eq:maxstab}).
\end{assumption}

Assumption~\ref{ass:regularity} adapts standard conditions \citep{firpo2007, cheng2024} to extreme-tail settings, relaxing the requirement that the density be bounded away from zero and providing a semi-parametric characterisation of tail behaviour that supports consistent extreme value approximations.

\subsection{From inverse estimating equations to tail calibration}\label{sec:tiee_integral}
The classical inverse estimating equation (IEE) framework identifies quantiles via the mapping $\theta \mapsto \tau_d(\theta)=\E[\1\{Y_d\leq\theta\}]$, which is strictly increasing, with $\theta_{d,\tau}$ defined by $\tau_d(\theta)=\tau$. 
Estimation proceeds by constructing a zero-mean estimating equation at level $\tau$ and inverting it pointwise.
This approach relies on local regularity. Indeed, the distribution function must vary sufficiently near $\theta_{d,\tau}$ and enough data must be available in its neighbourhood. These conditions fail in the tail, where the effective sample size collapses and $\tau_d(\theta)$ becomes nearly flat, rendering the inversion ill-conditioned and unstable.
The TIEE framework addresses this limitation by replacing pointwise inversion with a distribution-level representation that aggregates information across quantile levels.
Specifically, note that we can write~\eqref{eq:tau0_IEE} as
\begin{equation*}
    \tau_d(\theta) = \int_0^1\1\{Q_{Y_d}(p)\leq\theta\}\,dp,\quad \theta\in\mathbb{R},
\end{equation*}
which expresses the distribution function as the measure of quantile levels whose values lie below $\theta$.
This identity holds independently of local density behaviour and does not rely on pointwise inversion.
Motivated by this representation, we define the random integrated signal $S_d(W,\theta)$ and its population counterpart $S_d(\theta)$ as
\begin{align}\label{eq:Sd}
    S_d(W,\theta)
&:=
\int_0^1 s_{\mathrm{TIEE},d}(W,\theta,\zeta,p)\,dp,\nonumber\\
S_d(\theta)
&:=
\mathbb{E}\!\left[S_d(W,\theta)\right] = \int_0^1\E[s_{\mathrm{TIEE},d}(W,\theta,\zeta,p)]\,dp,
\end{align}
where $\zeta$ denotes nuisance parameters and the last equality in~\eqref{eq:Sd} is true under mild integrability conditions (e.g., $\E[\int_0^1|s_{\mathrm{TIEE},d}(W,\theta,\zeta,p)|\,dp] <\infty$).
The identifying requirement of the TIEE framework is the \emph{integrated unbiasedness condition}
\begin{equation}\label{eq:unbiasedness}
    S_d(\theta) = \int_0^1\E[s_{\mathrm{TIEE},d}(W,\theta,\zeta,p)]\,dp= \tau_d(\theta)\quad\text{for all } \theta\in\mathbb{R}.
\end{equation}

Under the condition in~\eqref{eq:unbiasedness}, identification is now achieved through a global probability balance rather than a local condition at a single quantile level.
The extreme quantile $\theta_{d,\tau}$ is therefore characterised as the unique solution to the integrated estimating equation
\begin{equation}\label{eq:int_est_eq}
    S_d(\theta) =\tau \quad\text{or equivalently}\quad \E[g_{\text{TIEE},d}(W,\theta,\tau,\zeta)] = 0,
\end{equation}
with $g_{\text{TIEE},d}(W,\theta,\tau,\zeta) = S_d(W,\theta) - \tau.$
This reformulation preserves the core logic of the IEE framework (quantiles are still obtained as roots of estimating equations) while stabilising the inversion by aggregating information across quantile levels.
As a result, the estimating equation remains well-defined even when empirical information near the target quantile is sparse. 
The integrated formulation naturally induces a decomposition into a data-rich body region and a tail region that requires extrapolation.
Indeed, $S_d(W,\theta)$ admits a natural decomposition into a body component (integrating over $p \leq u$) and a tail component (integrating over $p > u$), for some threshold $u \in (0,1)$.
We can then express the condition in~\eqref{eq:int_est_eq} as
$$\E\!\left[g_{\mathrm{TIEE},d}(W,\theta,\tau,\zeta)
\right] = \E\left[\int_0^1 \left\{s_{\mathrm{TIEE},d}(W, \theta_{d,\tau}, \zeta, p) - \tau_{\mathrm{eff}}\right\} \mathrm{d}p\right] = 0,$$
where $\tau_{\mathrm{eff}}=\frac{\tau-{p_u}}{1-{p_u}} = \Pr(Y_{d}\leq \theta_{d,\tau}\mid Y_d>u)$ denotes the quantile level relative to the tail distribution.

\subsection{Causal signal functions} \label{sec:tiee_signal}
The integrated estimating equation in Section~\ref{sec:tiee_integral} is agnostic to how the distribution of $Y_d$ is identified. Causal identification enters only through the choice of signal function, leaving the estimator’s structure unchanged and allowing flexibility across causal designs.
When the unbiasedness condition in~\eqref{eq:unbiasedness} holds, the signal $s_{\text{TIEE}}$ recovers the marginal distribution of $Y_d$ upon integration over $p$, and the integrated estimating equation in~\eqref{eq:int_est_eq} yields the target quantile. This condition can be achieved through different signal constructions; we describe two canonical examples relevant to our applications.

\noindent\textbf{Naïve signal under randomisation.}
In a randomised experiment, treatment assignment is independent of the potential outcomes. In this setting, a valid signal function is given by the indicator itself,
\begin{equation}\label{eq:signal_naive_rewrite}
s_{\text{TIEE}}(W,\theta,p)
\;=\;
\1\{Q_{Y_d}(p) \le \theta\}.
\end{equation}
This signal trivially satisfies the unbiasedness condition in~\eqref{eq:unbiasedness}, but it is not valid in observational studies where treatment assignment depends on covariates.

\noindent\textbf{Inverse probability weighted signal.}
In observational settings, causal identification can be achieved under unconfoundedness and overlap assumptions through inverse probability weighting. Let $\pi_d(\pmb{X};\zeta) = \Pr(D=d \mid \pmb{X})$ denote the propensity score. 
A valid TIEE signal is then given by
\begin{equation*}\label{eq:signal_ipw_rewrite}
s_{\text{TIEE}}(W,\theta,\zeta,p)
\;=\;
\frac{\1\{D=d\}}{\pi_d(\pmb{X};\zeta)}\,\1\{Q_{Y_d}(p)\le \theta\}.
\end{equation*}
Under standard causal assumptions, this signal satisfies~\eqref{eq:unbiasedness} and recovers the marginal distribution of $Y_d$ upon integration over $p$.

\subsection{Tail modelling and extrapolation}
\label{sec:tiee_tail}

The integrated estimating equation requires a model-assisted representation of the signal in the extreme tail, where data are scarce. When the target quantile exceeds a high threshold $u$, extrapolation becomes unavoidable. We address this by embedding an extreme value model directly into the estimating equation.
Specifically, under Assumption~\ref{ass:regularity}, exceedances above $u$ follow a generalised Pareto distribution, yielding the quantile representation
\begin{equation}
Q_{Y_d}(\tau) = u + \frac{\sigma_d}{\xi_d} \left[\left(\frac{1 - p_u}{1 - \tau}\right)^{\xi_d}- 1\right],
\end{equation}
for $Q_{Y_d}(\tau)>u$, where $p_u=\Pr(Y_d\le u)$. This motivates the effective tail probability $\tau_{\mathrm{eff}}=(\tau-p_u)/(1-p_u)$ introduced in Section~\ref{sec:tiee_integral}.

Within the TIEE framework, the GPD enters through the signal rather than as a plug-in estimator. The signal uses empirical indicators in the body and model-assisted counterparts in the tail. For the IPW signal in Section~\ref{sec:tiee_signal}, this yields
\[
s_d(W,\theta,\zeta)
=
\frac{\1\{D=d\}}{\pi_d(X;\zeta)}\,\1\{\tilde Y_d \le \theta\},
\]
where $\tilde Y_d$ is generated from the fitted GPD (e.g.\ $\tilde Y_d=Q_{Y_d}(U)$ with $U\sim \mathrm{Uniform}(0,1)$ or via transformed residuals). This preserves the probabilistic interpretation while enabling evaluation beyond the observed range. In practice, $\tilde Y_d$ is obtained via regression models on covariates (Sections~\ref{sec:simulation} and~\ref{sec:app}).

\subsection{Numerical solution for the TIEE estimator $\hat\theta_{d,\tau}$}
\label{sec:tiee_numerical}

The TIEE estimator $\hat\theta_{d,\tau}$ is defined as the solution to the empirical analogue of the integrated estimating equation introduced in Section~\ref{sec:tiee_integral}. Specifically, letting
\begin{equation}\label{eq:phi_n}
    \Phi_{d,n}(\theta)
=
\frac{1}{n}
\sum_{i=1}^n
g_{\mathrm{TIEE},d}(W_i,\theta,\hat\zeta,p)=\frac{1}{n}
\sum_{i=1}^n\int_0^1 s_{\mathrm{TIEE},d}(W,\theta,\zeta,p)\,dp,
\end{equation}
the estimator satisfies
$\Phi_n(\hat\theta_{d,\tau}) = 0.$
In practice, the integral over $p\in[0,1]$ is approximated using a grid $\{p_k\}_{k=0}^K$ with $K$ sufficiently large. The empirical estimating equation is then approximated by
\begin{equation}\label{eq:phi_n_discrete}
\frac{1}{n}
\sum_{i=1}^n
\sum_{k=1}^K
s_{\mathrm{TIEE},d}(W,\theta,\zeta,p)\,(p_k-p_{k-1})
= 0.
\end{equation}
We use $K=800$ for $n=1000$ and $K=2000$ for $n=5000$; sensitivity to $K$ is assessed in Appendix~\ref{app:sensitive_u}.

Direct root-finding is often unstable due to the non-smoothness induced by indicator functions, particularly in the tail. Following \cite{ma2022}, we instead solve an equivalent convex optimisation problem, defining an objective $L_{d,n}(\theta)$ whose subgradient is proportional to $\Phi_{d,n}(\theta)$. A convenient choice is
\begin{equation*}\label{eq:Ln_def}
L_{d,n}(\theta)
=
-
\sum_{i=1}^n
\sum_{k=1}^K
(p_k-p_{k-1})\,\tilde{S}_d(W_i,\theta,\hat\zeta,p_k)
+
\left|
R^\ast - \theta
\right|
\sum_{i=1}^n
\sum_{k=1}^K
\tau_{\mathrm{eff}}\,(p_k-p_{k-1}),
\end{equation*}
where $\tilde{S}_d(W,\theta,\zeta,p_k)$ is an antiderivative of the signal function with respect to $\theta$ at $p_k$, i.e., it satisfies
$\frac{\partial}{\partial \theta} \tilde{S}_d(W,\theta,\zeta,p_k) = s_{\mathrm{TIEE},d}(W,\theta,\zeta,p_k),$
and $R^\ast$ is a sufficiently large constant ensuring boundedness of the optimisation domain. The TIEE estimator is then obtained as
$\widehat\theta_{d,\tau} = \widehat{Q}_{Y_d}(\tau) = \arg\min_{\theta \in \Theta} L_{d,n}(\theta)$.
This convex $L_1$-type formulation ensures numerical stability and can be solved efficiently, even for extreme quantiles.


\section{Theoretical properties}\label{sec:property}

In this section, we establish the theoretical foundations of the TIEE framework. Our study centres on the TIEE estimating function $g_{\mathrm{TIEE},d}(W,\theta,\tau,\zeta)$ defined in~\eqref{eq:int_est_eq}, and its associated population integrated moment $\Phi_d(\theta)=\mathbb{E}[g_{\mathrm{TIEE},d}(W,\theta,\tau,\zeta)]$. 
Note that by construction, $\Phi_d(\theta) = S_d(\theta) - \tau$, where $S_d(\theta)=\mathbb{E}[S_d(W,\theta)]$ is the integrated signal defined in the second line of~\eqref{eq:Sd}.
Thus, $\Phi_d(\theta)$ is simply the centred version of $S_d(\theta).$
We assume that $\theta\in\Theta$, where $\Theta\subset\mathbb{R}$ is a compact parameter space.

\subsection{Identification, uniqueness, consistency and asymptotic normality}\label{sec:asymp_results}

Identification hinges on two crucial properties of the integrated moment function $g_{\mathrm{TIEE},d}(W,\theta,\tau,\zeta)$. 
First, that it correctly recovers the quantile condition when evaluated at $\theta_{d,\tau}$, and second, that it is strictly monotonic in $\theta$.

\begin{theorem}[Identification and Uniqueness]\label{thm:id}
Let $Y$ be a real-valued random variable with continuous marginal distribution function. Suppose the target quantile $\theta_{d,\tau}$ is uniquely defined by $F_{Y_d}(\theta_{d,\tau})=\tau$ for $\tau\in(0,1)$. Assume that for every $(W,\tau,\zeta)$ the function $g_{\mathrm{TIEE},d}(W,\theta,\tau,\zeta)$ is measurable and strictly increasing in $\theta$, and that the mapping $\Phi_d(\theta)=\E[g_{\mathrm{TIEE},d}(W,\theta,\tau,\zeta)]$ is well defined on $\Theta$ and differentiable in $\theta$. Then $\Phi_d(\theta_{d,\tau})=0$ and $\Phi_d(\theta)$ is strictly increasing in $\theta$. Consequently, the equation $\Phi_d(\theta)=0$ admits the unique solution $\theta=\theta_{d,\tau}$.

\end{theorem}

\begin{proof}
By the unbiasedness property of the signal function (see \eqref{eq:unbiasedness}), we have
\[
\Phi_d(\theta_{d,\tau})
=
\E\!\left[g_{\mathrm{TIEE},d}(W,\theta_{d,\tau},\tau,\zeta)\right]
=
\E\!\left[\int_0^1 \left\{ s(W,\theta_{d,\tau},\zeta,p)-\tau_{\text{eff}} \right\}\,\mathrm{d}p \right]
=
0.
\]
Moreover, $\Phi_d(\theta)$ is strictly increasing in $\theta$ because 
$g_{\mathrm{TIEE},d}(W,\theta,\tau,\zeta)$ is strictly increasing in $\theta$. 
Therefore, the equation $\Phi_d(\theta)=0$ admits at most one root. 
Since $\Phi_d(\theta_{d,\tau})=0$, this root must coincide with $\theta_{d,\tau}$.
\end{proof}

Having established that the true extreme quantile $\theta_{d,\tau}$ is uniquely identified as the root of the population integrated moment $\Phi_d(\theta) = 0$, we now turn to the convergence of its empirical counterpart.
The following Lemma is required to prove consistency, and its proof is deferred to Appendix~\ref{app:proofs}.

\begin{lemma}[Glivenko--Cantelli property of the integrated signal]\label{lemma_gc}
For $\theta \in \Theta$, let $S_{d}({W},\theta)$ be defined as in~\eqref{eq:Sd}, with $s_{\mathrm{TIEE},d}$ defined as in~\eqref{eq:signal_naive_rewrite}.
Then the class of functions $\mathcal{F}
\;=\;
\bigl\{\, S_d(\cdot,\theta) : \theta \in \Theta \,\bigr\}$
is Glivenko--Cantelli.
\end{lemma}

\begin{theorem}[Consistency]\label{thm:consistency}
Assume the conditions of Theorem~\ref{thm:id}. In addition, suppose that the covariate $\boldsymbol{X}$ has compact support, that for fixed $p$ the mapping $\theta \mapsto \1\{Q_{Y_d}(p)\leq\theta\}$ is monotonic and piecewise constant, and that the nuisance parameter estimates $\hat{\zeta}$ converge in probability to $\zeta$.
Then, the TIEE estimator $\hat{\theta}_{d,\tau}$ that solves the empirical estimating equation $\Phi_{d,n}(\theta) = 0$ is consistent for the true extreme quantile $\theta_{d,\tau}$, i.e.,
$\hat{\theta}_{d,\tau} \xrightarrow{P} \theta_{d,\tau}$.
\end{theorem}

\begin{proof}
Consistency follows from the general theory of Z-estimators
\citep[Theorem~5.9]{van2000asymptotic}. Unique identification holds by
Theorem~\ref{thm:id}. For uniform convergence, Lemma~\ref{lemma_gc}
implies that the class $\mathcal F=\{S_d(\cdot,\theta):\theta\in\Theta\}$
is Glivenko--Cantelli. Since
$g_{\mathrm{TIEE},d}(W,\theta,\tau)=S_d(W,\theta)-\tau$, the class
$\{g_{\mathrm{TIEE},d}(\cdot;\theta,\tau):\theta\in\Theta\}$ is also
Glivenko--Cantelli, yielding
\[
\sup_{\theta\in\Theta}
\left|
\frac{1}{n}\sum_{i=1}^n g_{\mathrm{TIEE},d}(W_i,\theta,\tau)
-
\E[g_{\mathrm{TIEE},d}(W,\theta,\tau)]
\right|
\xrightarrow{P}0 .
\]
Replacing $\zeta$ by $\hat\zeta$ introduces a uniformly $o_P(1)$ error,
so $\sup_{\theta\in\Theta}|\Phi_{d,n}(\theta)-\Phi_d(\theta)|\xrightarrow{P}0$.
Since $\Phi_d$ has a unique zero at $\theta_{d,\tau}$, it follows that
$\hat{\theta}_{d,\tau}\xrightarrow{P}\theta_{d,\tau}$.
\end{proof}
Theorem~\ref{thm:consistency} establishes that the TIEE estimator is consistent even when the target quantile $\tau$ is extreme and lies beyond the range of observed data. This result is non-trivial because it holds despite the vanishing density at $\theta_{d,\tau}$, a consequence of the EVT-based regularisation embedded in the integrated moment condition.
Building on this result, we derive the asymptotic distribution of the TIEE under suitable regularity conditions. The complete proof, including the treatment of nuisance parameter estimation via the functional delta method \citep{van2000asymptotic}, is given in Appendix~\ref{app:proofs}.

\begin{theorem}[Asymptotic Normality]\label{thm:asymptotic}
Assume the conditions of Theorems~\ref{thm:id} and~\ref{thm:consistency}. Suppose further that $\Phi_d(\theta)$ is continuously differentiable in a neighbourhood of $\theta_{d,\tau}$ with $\Phi_d'(\theta_{d,\tau})\neq0$, that $\sqrt{n}(\hat{\zeta}-\zeta)=O_p(1)$ and $\Phi_{d,n}(\theta)$ is pathwise differentiable in $\zeta$, and that $\sigma_{d,\tau}^2:=\mathrm{Var}[g_{\mathrm{TIEE},d}(W,\tau,\theta_{d,\tau},\zeta)]$ is finite. Then
\begin{equation}\label{eq:asymptotic_normality}
\sqrt{n}(\hat{\theta}_{d,\tau}-\theta_{d,\tau})
\xrightarrow{d}
\mathcal{N}\!\left(0,V_{d,\tau}\right),
\qquad
V_{d,\tau}=\frac{\sigma_{d,\tau}^2}{[\Phi_d'(\theta_{d,\tau})]^2}.
\end{equation}
\end{theorem}
The asymptotic variance in~\eqref{eq:asymptotic_normality} highlights several key features of the TIEE estimator. 
First, unlike empirical quantiles, it remains finite even in extreme regimes where $f_Y(\theta_{d,\tau}) \to 0$ as $\tau \to 1$. 
This reflects the EVT-based regularisation induced by integrating the estimating equation over quantile levels, which stabilises tail inference by borrowing strength from neighbouring quantiles.
Second, the construction is also compatible with classical extreme value limits, yielding asymptotic normality under intermediate sequence regimes where empirical quantiles typically fail.
Third, the variance $V_{d,\tau}$ depends on the derivative of the population moment function, which, for the signal functions in Section~\ref{sec:tiee_signal}, coincides with the density of the potential outcome at the target quantile, $\Phi_d'(\theta)=f_{Y_d}(\theta)$.
Finally, $\sigma_{d,\tau}^2$ also captures uncertainty from estimating nuisance parameters $\zeta$. When the propensity score is correctly specified, the IPW-based TIEE estimator is asymptotically efficient.

\subsection{Inference for the extreme quantile treatment effect $\delta(\tau)$}\label{sec:eqte_inference}

Based on the results from Section~\ref{sec:asymp_results}, we derive an expression for the asymptotic distribution of the extreme quantile treatment effect $\delta(\tau) = \theta_{1,\tau} - \theta_{0,\tau}$ (as defined in~\eqref{eq:qte}).
Standard Z-estimation arguments yield the joint asymptotic normality
\[
\sqrt{n}
\begin{pmatrix}
\hat{\theta}_{1,\tau} - \theta_{1,\tau} \\
\hat{\theta}_{0,\tau} - \theta_{0,\tau}
\end{pmatrix}
\xrightarrow{d}
\mathcal{N}\!\left(0, D^{-1} \Sigma_\Phi D^{-1}\right),
\]
where $D = \mathrm{diag}\!\big(\Phi'_1(\theta_{1,\tau}), \Phi'_0(\theta_{0,\tau})\big)$ and $\Sigma_\Phi = \text{Var}\!\big((g_1(W,\theta_{1,\tau}),\, g_0(W,\theta_{0,\tau}))^\top\big)$ with componentes $\Sigma_{dd'}$.
It follows that $\sqrt{n}\big(\hat{\delta}(\tau) - \delta(\tau)\big) \xrightarrow{d}\mathcal{N}\!\left(0,\; \sigma_\delta^2\right),$ where
\begin{equation}\label{eq:eqte_variance}
    \sigma_\delta^2
=
\frac{\Sigma_{11}}{[\Phi'_1(\theta_{1,\tau})]^2}
+
\frac{\Sigma_{00}}{[\Phi'_0(\theta_{0,\tau})]^2}
-
\frac{2\Sigma_{10}}{\Phi'_1(\theta_{1,\tau})\Phi'_0(\theta_{0,\tau})}.
\end{equation}
A consistent estimator $\hat{\sigma}_\delta^2$ is obtained via sample analogues (see Appendix~\ref{app:variance}).
The terms $\Phi'_d(\theta_{d,\tau})$ can be estimated via numerical differentiation or density estimation.

\section{Simulation Study} \label{sec:simulation}

To evaluate the finite-sample performance of the proposed TIEE framework, we conduct a simulation study with three objectives: (i) to compare the inverse probability weighted TIEE estimator (TIEE-IPW) with existing methods, including the Zhang--Firpo estimator \citep{zhang2018}, the Hill causal estimator \citep{deuber2024estimation}, and the Pickands estimator \citep{zhang2018supplement}; (ii) to examine performance across different tail regimes (heavy- versus light-tailed distributions); and (iii) to assess robustness of TIEE-IPW under propensity score misspecification.

We consider sample sizes $n\in\{1000,5000\}$ and report bias and mean squared error (MSE) based on $1000$ Monte Carlo replications. The extreme quantile level is defined through the tail probability $1-\tau_n$, allowing us to study three regimes commonly encountered in extreme quantile estimation. Specifically, we set $(1-\tau_n)\in\{5/n,\,1/n,\,5/(n\log n)\}$, corresponding to intermediate ($n(1-\tau_n)\to\infty$), moderately extreme ($n(1-\tau_n)\to d>0$), and very extreme ($n(1-\tau_n)\to0$) regimes.

\subsection{Data-Generating Process}

The data-generating process is identical across scenarios. We generate $n$ i.i.d.\ observations $W(i)=(Y(i),D(i),\pmb{X}(i))$, where $\pmb{X}(i)=(1,X(i))^\top$ and $X(i)\sim\text{Uniform}(-1,1)$. The propensity score for the binary treatment $D=1$ is
\begin{equation}\label{eq:simulation_ps}
\pi(\pmb{x})=\mathbb{P}(D=1\mid\pmb{X}=\pmb{x})=0.5x^2+0.25.
\end{equation}
Treatment assignments are generated as $D(i)\mid\pmb{X}(i)\sim\text{Ber}(\pi(\pmb{X}(i)))$, and the potential outcomes $Y_1$ and $Y_0$ follow the tail-regime specifications described below.

Following \cite{deuber2024estimation}, the intermediate quantile level associated with the GPD threshold is set to $(1-p_u)=k/n$ with $k=n^{0.65}$. Although not optimal in all settings, this choice provides a common benchmark. The tail of the reconstructed potential outcome distribution is obtained via GPD regression on exceedances above the threshold, with GPD parameters modelled as functions of covariates.
Additional sensitivity analyses for the threshold $u$ and grid size $K$ are provided in Appendix~\ref{app:sensitive}.

\subsection{Heavy-tailed scenarios}\label{sec:simulation_heavy}

This set of scenarios focuses on cases where the underlying potential outcome distributions $F_{Y_d}$ belong to the Fréchet maximum domain of attraction, characterised by heavy tails. The three scenarios vary the heterogeneity of the EVIs between the treatment and control groups. 
Following \cite{deuber2024estimation}, the potential outcomes are generated from the following models:
\[
\begin{aligned}
M_1^{\text{(H)}}:\;&
\begin{cases}
Y_1 = 5S(1+X),\\
Y_0 = S(1+X),
\end{cases}
\;
M_2^{\text{(H)}}:\;&
\begin{cases}
Y_1 = C_2\exp(X),\\
Y_0 = C_3\exp(X),
\end{cases}
\;
M_3^{\text{(H)}}:\;&
\begin{cases}
Y_1 = P_{1.75+X,2},\\
Y_0 = P_{1.75+X,1},
\end{cases}
\end{aligned}
\]
where $X$ is defined as above, $S$ follows a Student-$t$ distribution with 3 degrees of freedom, $C_s$ is Fréchet distributed with shape parameter $s$, location 0, and scale 1, and $P_{a,b}$ is Pareto distributed with shape parameter $a$ and scale $b$. The associated EVIs for each scenario are $\gamma_1 = \gamma_0 = 1/3$ for model $M_1^{\text{(H)}}$, $\gamma_1 = 1/2$ and $\gamma_0 = 1/3$ for model $M_2^{\text{(H)}}$, and $\gamma_1 = \gamma_0 = 4/7$ for model $M_3^{\text{(H)}}$.

\subsection{Light-tailed scenarios}\label{sec:simulation_light}

Here we focus on scenarios where the underlying distribution $F_{Y_d}$ belong to the Gumbel or Weibull maximum domain of attraction, corresponding to light-tailed distributions. 
Here we expect EVT-based estimators that primarily target heavy-tailed data (such as the Hill estimator method) to not perform very well. 
The potential outcomes are generated from the following three models:
\[
\begin{aligned}
M_1^{\text{(L)}}:
\begin{cases}
Y(1) = 5S(1+X),\\
Y(0) = S(1+X),
\end{cases}
\;
M_2^{\text{(L)}}:
\begin{cases}
Y(1) = C_1\exp(X),\\
Y(0) = C_2\exp(X),
\end{cases}
\;
M_3^{\text{(L)}}:
\begin{cases}
Y(1) = W_{2+X,\,2},\\
Y(0) = W_{3+2X,\,1}.
\end{cases}
\end{aligned}
\]
where $X$ is defined as above, $S$ follows a standard normal distribution, $C_1 \sim \text{Exp}(1)$, $C_2 \sim \text{Exp}(2)$, and $W_{a,b}$ denotes a Weibull distribution with shape parameter $a$ and scale $b$. 
The associated EVIs for each scenario are $\gamma_1 = \gamma_0 = 0$ for models $M_1^{\text{(L)}}$ and $M_2^{\text{(L)}}$ (corresponding to the Gumbel domain of attraction), and $\gamma_1 = -1/(2+X)$ and $\gamma_0=-1/(3+2X)$ for model $M_3^{\text{(L)}}$.

\subsection{Results}
Table~\ref{tab:heavy_tail} summarises estimation performance under both heavy- and light-tailed scenarios for $n=1000$. 
In the heavy-tailed case, TIEE-IPW consistently achieves lower bias and variance than the Zhang-Firpo and Pickands estimators across all tail regimes. 
While the Causal Hill estimator performs reasonably well in the intermediate regime ($p_n = 5/n$), TIEE-IPW remains stable across all regimes, including the very extreme tail ($p_n = 5/(n\log n)$), where its advantage is most pronounced. 
The Pickands estimator exhibits substantially larger bias and variance throughout, particularly in the most extreme regime, reflecting its known finite-sample variability. 
The Zhang-Firpo estimator performs competitively in the intermediate regime but deteriorates in Scenario $M_2^{\text{(H)}}$, where EVI heterogeneity is strongest.
In the light-tailed case, all estimators perform markedly better, as reflected by lower MSE values. 
The relative advantage of TIEE-IPW persists but is less pronounced, consistent with the easier estimation of light-tailed extremes. 
The Pickands estimator again performs worst, with inflated MSE, especially in the very extreme tail regime.
Similar results for $n=5000$ are shown in Appendix~\ref{app:n5000}.

Figure~\ref{fig:simu_heavy_light} reports the empirical coverage of the 90\% intervals across both tail designs. 
In the heavy-tailed case, both the Causal Hill and TIEE estimators exhibit stable coverage, consistent with their compatibility with Pareto-type tails, while Zhang systematically under-covers, with deviations increasing in the extreme tail. 
Overall, TIEE achieves the most accurate and robust finite-sample coverage. Results for the Pickands estimator are omitted due to its highly unstable interval estimates.
In the light-tailed case, TIEE maintains coverage close to the nominal level across all designs and regimes. 
Zhang continues to under-cover, particularly at more extreme quantiles. 
The Causal Hill estimator shows design-dependent behaviour, with substantial under-coverage in $M_1^{\text{(L)}}$ and $M_3^{\text{(L)}}$, but near-nominal performance in $M_2^{\text{(L)}}$. 
Overall, TIEE provides the most uniform and reliable uncertainty quantification, remaining robust in settings with mild tails and limited extreme observations.

\begin{table}[H]
\centering
\scriptsize
\caption{Estimation performance (Bias and MSE) under heavy- and light-tailed scenarios for sample size $n=1000$.}
\label{tab:heavy_tail}
\begin{tabular}{llrrrrrr}
\toprule
\multicolumn{8}{c}{Heavy-tailed scenarios} \\
\midrule
 & & \multicolumn{2}{c}{\textbf{Scenario $M_1^{\text{(H)}}$}} & \multicolumn{2}{c}{\textbf{Scenario $M_2^{\text{(H)}}$}} & \multicolumn{2}{c}{\textbf{Scenario $M_3^{\text{(H)}}$}} \\
\cline{3-4}\cline{5-6}\cline{7-8}
$1-\tau_n$ & Method & Bias & MSE & Bias & MSE & Bias & MSE \\
\midrule
\multirow{4}{*}{$5/(n \log n)$}
 & Causal Hill  & 45.982 & 4828.099  & 6.370   & 1355.930   & 26.759  & 3652.382   \\
 & Pickands     & -52.322 & 5694.769 & 42.410  & 56223.057  & 63.741  & 86214.119  \\
 & Zhang-Firpo  & -3.243 & 2840.382  & 0.847   & 5395.063   & -9.793  & 6191.232   \\
 & TIEE-IPW     & -6.792 & \textbf{1513.337} & -16.460 & \textbf{581.743} & -5.879 & \textbf{544.757} \\
\cline{1-8}
\multirow{4}{*}{$1/n$}
 & Causal Hill  & 28.305  & 2329.556   & 4.388   & 844.775     & 15.660  & 1872.560    \\
 & Pickands     & -36.222 & 6194.962   & 59.070  & 110234.161  & 58.364  & 101126.968  \\
 & Zhang-Firpo  & 4.117   & 2846.812   & 8.653   & 5469.215    & -2.782  & 6103.073    \\
 & TIEE-IPW     & -4.997 & \textbf{1025.129} & -10.776 & \textbf{390.789} & -4.072 & \textbf{390.293} \\
\cline{1-8}
\multirow{4}{*}{$5/n$}
 & Causal Hill  & 4.496   & 135.273  & 0.539  & 55.963   & 2.142  & 79.982   \\
 & Pickands     & -19.168 & 899.346  & 6.520  & 1714.555 & 4.988  & 1268.937 \\
 & Zhang-Firpo  & 1.194   & 174.980  & 1.183  & 131.197  & 3.330  & 1128.588 \\
 & TIEE-IPW     & -1.193 & \textbf{101.772} & -2.066 & \textbf{52.606} & -1.680 & \textbf{68.289} \\
\midrule
\multicolumn{8}{c}{Light-tailed scenarios} \\
\midrule
 & & \multicolumn{2}{c}{\textbf{Scenario $M_1^{\text{(L)}}$}} & \multicolumn{2}{c}{\textbf{Scenario $M_2^{\text{(L)}}$}} & \multicolumn{2}{c}{\textbf{Scenario $M_3^{\text{(L)}}$}} \\
\cline{3-4}\cline{5-6}\cline{7-8}
$1-\tau_n$ & Method & Bias & MSE & Bias & MSE & Bias & MSE \\
\midrule
\multirow{4}{*}{$5/(n \log n)$}
 & Causal Hill  & 13.488  & 244.970 & 4.372  & 49.431   & 0.535  & 0.601  \\
 & Pickands     & -15.687 & 482.247 & 5.818  & 3791.332 & -1.827 & 15.487 \\
 & Zhang-Firpo  & -0.337  & 14.195  & 0.475  & 12.403   & -0.347 & 0.307  \\
 & TIEE-IPW     & -1.286 & \textbf{10.593} & -0.153 & \textbf{9.121} & -0.097 & \textbf{0.252} \\
\cline{1-8}
\multirow{4}{*}{$1/n$}
 & Causal Hill  & 10.916  & 163.065 & 3.459  & 32.668   & 0.441  & 0.447  \\
 & Pickands     & -15.195 & 408.884 & 3.743  & 2018.179 & -1.788 & 12.386 \\
 & Zhang-Firpo  & 0.406   & 14.245  & 0.847  & 12.895   & -0.266 & 0.257  \\
 & TIEE-IPW     & -1.144 & \textbf{8.332} & -0.109 & \textbf{6.573} & -0.089 & \textbf{0.196} \\
\cline{1-8}
\multirow{4}{*}{$5/n$}
 & Causal Hill  & 2.563   & 12.662  & 0.734  & 3.189   & 0.105  & 0.083 \\
 & Pickands     & -11.699 & 183.880 & -0.564 & 106.240 & -1.438 & 4.721 \\
 & Zhang-Firpo  & -0.090  & 4.283   & 0.058  & 2.774   & -0.015 & 0.089 \\
 & TIEE-IPW     & -0.488 & \textbf{2.484} & 0.173 & \textbf{1.334} & -0.063 & \textbf{0.050} \\
\bottomrule
\end{tabular}
\end{table}

\begin{figure}[htbp]
    \centering
    \includegraphics[width=\textwidth]{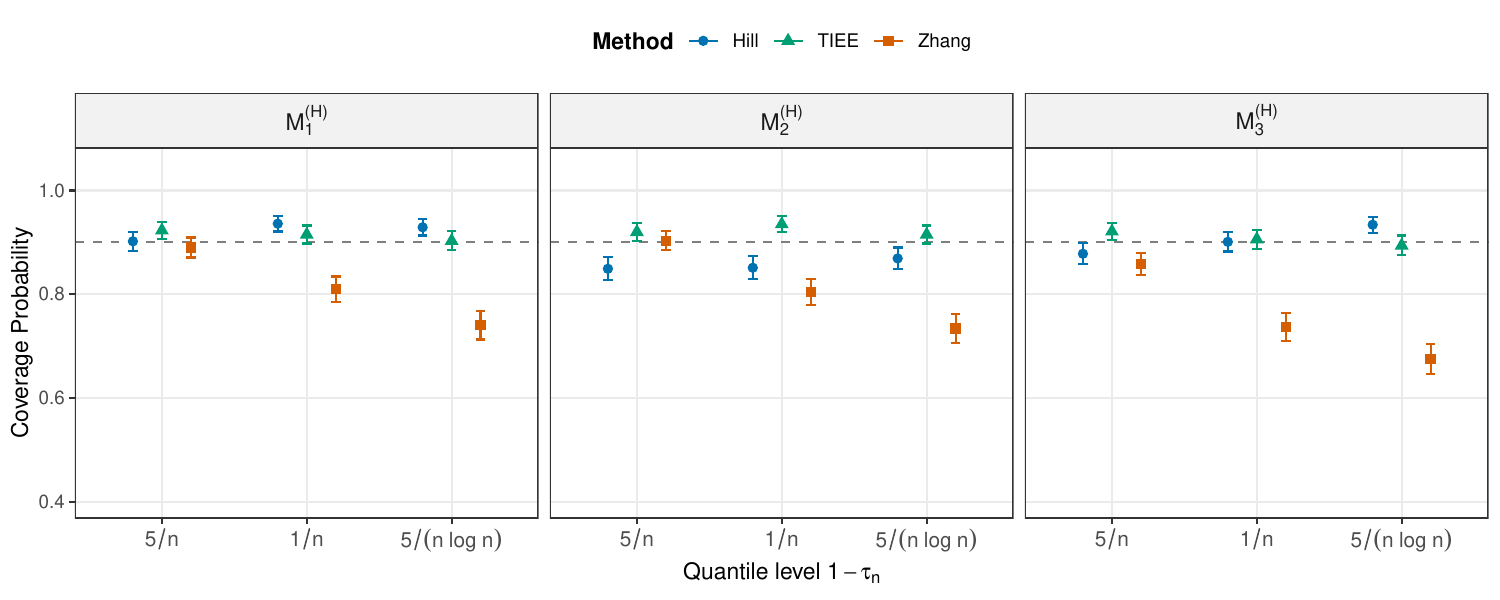}
    \includegraphics[width=\textwidth]{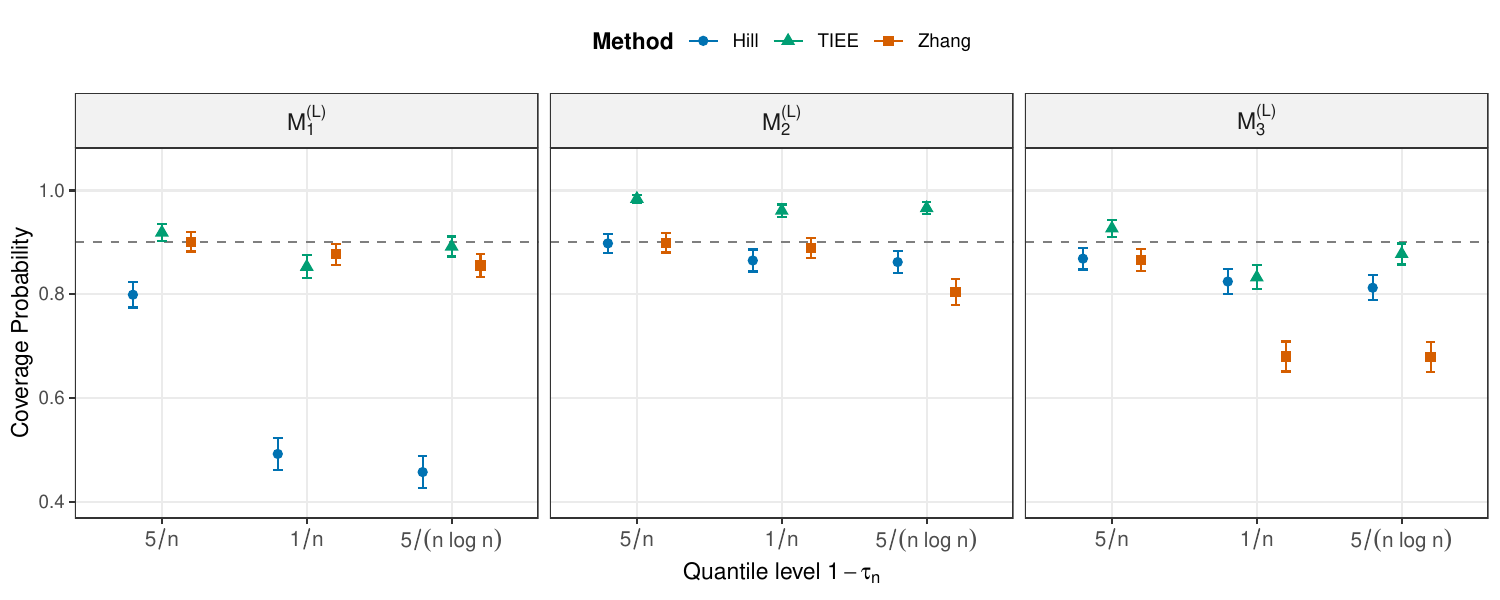}
    \caption{Coverage rates of 90\% confidence intervals for the extreme quantile treatment effect. 
    The comparison includes Zhang, Hill, and TIEE estimators with varying sample fractions ($5/n$, $1/n$, and $5/(n \log n)$). The dashed horizontal line represents the nominal confidence level. The top row shows results for the heavy-tailed scenarios, while the bottom row shows results for the light-tailed scenarios.
    } 
    \label{fig:simu_heavy_light}
\end{figure}

\subsubsection{Robustness to model misspecification\label{sec:robust}} 
In the extreme quantile regime, the outcome model is inherently misspecified, as tail behaviour cannot be learned directly from the data, so classical double robustness does not strictly apply. 
Nevertheless, simulations in Sections~\ref{sec:simulation_heavy} and~\ref{sec:simulation_light} show that TIEE remains robust to tail-model misspecification, despite imposing a GPD on non-GPD data.
We further investigate robustness to propensity score misspecification by fitting incorrect models.
Specifically, we consider three misspecification scenarios: (1) omitting the quadratic term in~\eqref{eq:simulation_ps}, (2) using a linear model when the true model is quadratic, and (3) introducing spurious interaction terms.
Table~\ref{tab:misspecify} shows that the EQTE estimator $\delta(\tau_n)$ maintains small bias and valid inference under these perturbations.
\begin{table}[htbp]
\centering
\caption{Bias and MSE of the TIEE estimator of the EQTE $\delta(\tau_n)$ under different propensity score misspecifications. The ``True Value'' refers to the EQTE computed from the underlying data-generating distributions via large-sample Monte Carlo approximation.}
\label{tab:misspecify}
\scriptsize
\begin{tabular}{cc l rr}
\toprule
$1-\tau_n$ & True Value & Propensity Score Model & Bias & MSE \\
\midrule
\multirow{5}{*}{$5/(n \log n)$} & \multirow{5}{*}{45.39} 
  & True Model (Quadratic) & -5.438 & 374.19 \\ 
  & & Polynomial Basis & -5.601 & 397.00 \\ 
  & & Misspecified Form (Linear) & -3.903 & 386.64 \\ 
  & & Misspecified Link (Logit) & -3.269 & 551.88 \\ 
  & & Spurious Covariates & -3.472 & 581.16 \\ 
\midrule
\multirow{5}{*}{$1/n$} & \multirow{5}{*}{38.38} 
  & True Model (Quadratic) & -3.837 & 275.01 \\ 
  & & Polynomial Basis & -3.838 & 283.39 \\ 
  & & Misspecified Form (Linear) & -2.516 & 291.51 \\ 
  & & Misspecified Link (Logit) & -2.199 & 361.22 \\ 
  & & Spurious Covariates & -2.103 & 370.79 \\ 
\midrule
\multirow{5}{*}{$5/n$} & \multirow{5}{*}{16.53} 
  & True Model (Quadratic) & -1.338 & 29.53 \\ 
  & & Polynomial Basis & -1.354 & 29.33 \\ 
  & & Misspecified Form (Linear) & -1.009 & 31.11 \\ 
  & & Misspecified Link (Logit) & -1.287 & 29.74 \\ 
  & & Spurious Covariates & -1.285 & 29.89 \\ 
\bottomrule
\end{tabular}
\end{table}
Taken together, these results show that TIEE remains numerically robust under misspecification of both the tail and propensity score models. While not formally doubly robust, this practical robustness is valuable in applications. 

\section{Data application}\label{sec:app}
Extreme event attribution (EEA) typically compares factual and counterfactual climates using model ensembles, either via changes in exceedance probabilities \citep{van2021pathways} or storyline approaches conditioning on atmospheric circulation~\citep{shepherd2016common}. While effective, these methods do not directly target changes in extreme quantiles.

We adopt a causal observational framework that complements EEA by focusing on extreme quantile treatment effects. Conditioning on large-scale circulation aligns with storyline attribution while enabling estimation directly from historical data.

We analyse the EEAR-Clim dataset~\citep{bongiovanni2024eear}, using seasonal 7-day maximum precipitation in the Austrian Alps as the outcome. The treatment contrasts modern (1995–2020) and historical (1950–1980) periods, targeting the $\tau=0.99$ quantile ($\approx$100-year return level under stationarity).

Two main challenges arise: data sparsity in the tail and confounding by atmospheric circulation. Valid attribution therefore requires separating thermodynamic effects from dynamical variability~\citep{shepherd2016common}. In our framework, circulation acts as a confounder, affecting both the climate regime and extreme precipitation. We address this by constructing a confounder set $\pmb{X}$ from ERA5 reanalysis data (ECMWF; \citealt{hersbach2020era5}\footnote{\textbf{Data}: Copernicus Climate Change Service (C3S) (2023). ERA5 hourly data on single levels [DOI: 10.24381/cds.adbb2d47] and pressure levels [DOI: 10.24381/cds.bd0915c6]. Accessed Sept. 2025.}), targeting the dynamical state of the atmosphere. Covariates include synoptic-scale circulation variables such as sea-level pressure (SLP) whcih characterises surface pressure systems, geopotential height at 500 hPa ($Z_{500}$) describing mid-tropospheric circulation, and wind intensity at 850 hPa ($\text{Wind}_{850}$), capturing lower-tropospheric flow.
Additionally, we include large-scale teleconnection indices (NAO and AO).
We explicitly exclude thermodynamic variables such as local temperature, which lie on the causal pathway and would induce over-adjustment, as well as slowly varying oceanic indices (e.g., AMO, PDO), whose persistence may violate the positivity assumption. 
The propensity score $\pi(\pmb{X})=\mathbb{P}(D=1|\pmb{X})$ is estimated via the follwoing logistic regression model 
\[
\log\left(\frac{\pi(\pmb{X})}{1-\pi(\pmb{X})}\right)
=
\beta_0 + \sum_{r=1}^2 \alpha_r Z_{500}^r + \sum_{r=1}^2 \gamma_r \text{SLP}^r + \delta \text{Wind}_{850} + \eta \text{NAO} + \sum_{r=1}^4 \lambda_r \text{AO}^r.
\]
This flexible specification ensures effective balancing of circulation patterns across periods, allowing us to isolate the thermodynamic contribution of anthropogenic forcing.

\subsection{Results}

We applied the TIEE, Causal Hill, and Zhang-Firpo estimators to the 23 stations listed in Table~\ref{tab:app}. 
Table~\ref{tab:app} addresses the attribution question by quantifying changes in the 99\% precipitation quantile under recent warming, conditional on comparable circulation states. 
While unadjusted and alternative estimators yield mixed and often inconclusive results, TIEE consistently identifies positive shifts in the upper tail across stations, indicating that anthropogenic warming primarily intensifies extreme precipitation rather than altering its frequency.
\begin{table}
\centering
\caption{Extreme quantile treatment effect estimates (99\% quantile) for 23 Austrian stations. Bold values in the TIEE column indicate statistical significance at the 10\% level.
\label{tab:app}}
\centering
\scriptsize{
\begin{tabular}[t]{lllll}
\toprule
Station & Empirical & Causal Hill & Zhang-Firpo & TIEE\\
\midrule
Aspang & -4.92 & -1.81 [-6.53, \phantom{-}2.90] & -3.70 [-10.05, \phantom{-}4.73] & -1.25 [-3.60, \phantom{-}1.11]\\
Bruckmur & \phantom{-}5.65 & \phantom{-}1.95 [-1.98, \phantom{-}5.89] & \phantom{-}4.20 [-5.07, \phantom{-}12.70] & \textbf{\phantom{-}3.01 [\phantom{-}0.92, \phantom{-}5.09]}\\
Deutschlandsberg & -2.67 & -4.07 [-8.31, \phantom{-}0.17] & -4.90 [-14.91, \phantom{-}3.26] & \textbf{-3.69 [-6.34, -1.04]}\\
Doellach & \phantom{-}5.22 & \phantom{-}2.85 [-0.07, \phantom{-}5.78] & \phantom{-}4.30 [-0.86, \phantom{-}11.16] & \textbf{\phantom{-}3.60 [\phantom{-}1.90, \phantom{-}5.30]}\\
Feuerkogel & \phantom{-}0.82 & \phantom{-}6.13 [-0.15, \phantom{-}12.42] & \phantom{-}7.90 [\phantom{-}2.64, \phantom{-}14.76] & \phantom{-}1.30 [-1.89, \phantom{-}4.49]\\
\addlinespace
Freistadt & \phantom{-}3.21 & \phantom{-}5.04 [\phantom{-}0.51, \phantom{-}9.57] & \phantom{-}4.30 [-1.51, \phantom{-}10.27] & \textbf{\phantom{-}5.41 [\phantom{-}2.84, \phantom{-}7.98]}\\
Graz Universitaet & -1.62 & -2.34 [-6.33, \phantom{-}1.66] & -3.10 [-10.90, \phantom{-}8.90] & -1.53 [-3.89, \phantom{-}0.82]\\
Hohenau & -0.25 & \phantom{-}2.32 [-1.21, \phantom{-}5.85] & \phantom{-}1.90 [-4.82, \phantom{-}6.78] & \phantom{-}0.47 [-1.41, \phantom{-}2.36]\\
Kremsmuenster Tawes & -3.84 & -2.05 [-6.58, \phantom{-}2.48] & -0.50 [-4.95, \phantom{-}3.42] & \textbf{-3.61 [-6.05, -1.17]}\\
Landeck & \phantom{-}1.45 & \phantom{-}0.94 [-1.45, \phantom{-}3.33] & -0.20 [-2.72, \phantom{-}2.93] & \textbf{\phantom{-}2.97 [\phantom{-}1.52, \phantom{-}4.43]}\\
\addlinespace
Mariazellst. Sebastian Flugfeld & -0.49 & \phantom{-}2.39 [-2.58, \phantom{-}7.36] & \phantom{-}0.10 [-3.42, \phantom{-}2.40] & \phantom{-}1.24 [-1.15, \phantom{-}3.63]\\
Mayrhofen & \phantom{-}1.34 & \phantom{-}0.44 [-2.60, \phantom{-}3.48] & \phantom{-}1.20 [-6.11, \phantom{-}8.63] & \phantom{-}0.68 [-0.94, \phantom{-}2.30]\\
Muerzzuschlag & \phantom{-}9.24 & \phantom{-}8.55 [\phantom{-}4.86, \phantom{-}12.24] & \phantom{-}6.70 [\phantom{-}0.51, \phantom{-}11.46] & \textbf{\phantom{-}8.99 [\phantom{-}7.00, \phantom{-}10.99]}\\
Patscherkofel & \phantom{-}3.11 & \phantom{-}1.61 [-1.32, \phantom{-}4.54] & \phantom{-}2.50 [-0.02, \phantom{-}5.86] & \textbf{\phantom{-}1.96 [\phantom{-}0.33, \phantom{-}3.60]}\\
Puchbergschneeberg & \phantom{-}6.27 & \phantom{-}8.18 [\phantom{-}3.66, \phantom{-}12.70] & \phantom{-}8.30 [-1.06, \phantom{-}16.81] & \textbf{\phantom{-}7.25 [\phantom{-}4.89, \phantom{-}9.62]}\\
\addlinespace
Retzwindmuehle & \phantom{-}6.88 & \phantom{-}6.17 [\phantom{-}3.12, \phantom{-}9.22] & \phantom{-}7.00 [\phantom{-}4.14, \phantom{-}9.23] & \textbf{\phantom{-}5.43 [\phantom{-}3.80, \phantom{-}7.07]}\\
Reutte & \phantom{-}4.51 & \phantom{-}8.94 [\phantom{-}3.49, \phantom{-}14.40] & \phantom{-}9.70 [-7.58, \phantom{-}22.47] & \textbf{\phantom{-}5.46 [\phantom{-}2.64, \phantom{-}8.29]}\\
Schoppernau & \phantom{-}7.25 & \phantom{-}8.42 [\phantom{-}3.04, \phantom{-}13.81] & \phantom{-}11.00 [\phantom{-}0.28, \phantom{-}21.58] & \textbf{\phantom{-}5.41 [\phantom{-}2.86, \phantom{-}7.96]}\\
Schroecken & \phantom{-}3.05 & \phantom{-}4.00 [-0.91, \phantom{-}8.90] & \phantom{-}4.00 [-4.87, \phantom{-}20.37] & \textbf{\phantom{-}3.08 [\phantom{-}0.47, \phantom{-}5.68]}\\
St. Jakobdef. & \phantom{-}0.41 & -2.32 [-5.74, \phantom{-}1.11] & -0.40 [-6.15, \phantom{-}3.68] & -0.24 [-2.15, \phantom{-}1.67]\\
\addlinespace
St. Poeltenlandhaus & \phantom{-}4.23 & \phantom{-}3.87 [-1.33, \phantom{-}9.07] & \phantom{-}5.20 [-3.28, \phantom{-}12.88] & \phantom{-}1.65 [-0.92, \phantom{-}4.21]\\
Stift Zwettl & \phantom{-}4.45 & \phantom{-}3.44 [-0.53, \phantom{-}7.41] & \phantom{-}5.30 [\phantom{-}1.74, \phantom{-}9.44] & \textbf{\phantom{-}2.24 [\phantom{-}0.09, \phantom{-}4.39]}\\
Waizenkirchen & -4.84 & -3.10 [-6.83, \phantom{-}0.63] & -1.30 [-7.78, \phantom{-}4.26] & \textbf{-4.88 [-6.90, -2.87]}\\
\bottomrule
\end{tabular}}
\end{table}
The results highlight substantial discrepancies between naive and causally adjusted estimates. 
Unadjusted empirical comparisons frequently suggest negative effects (e.g., Aspang $-4.92$~mm, Waizenkirchen $-4.84$~mm), which we attribute to circulation-driven variability rather than true thermodynamic responses. 
In contrast, TIEE corrects these distortions and reveals coherent positive signals.

TIEE also demonstrates greater power and precision. 
At Landeck, it detects a significant increase ($\hat{\delta}_{\mathrm{TIEE}}(0.99)=2.97$~mm), whereas Zhang-Firpo suggests a decrease ($-0.20$~mm), illustrating how inefficient estimators can obscure the signal. 
More broadly, Zhang-Firpo produces very wide intervals (e.g., $\approx 18$~mm at Bruckmur), while TIEE yields substantially narrower intervals (typically $\approx 4$--$5$~mm) and identifies significant effects at stations such as Bruckmur, Freistadt, and Schoppernau. 
The Causal Hill estimator generally recovers the correct direction but suffers from large variance, often failing to detect significance.

Aggregating across stations, TIEE reveals a coherent regional signal. 
While empirical estimates vary widely due to local dynamical noise, TIEE produces consistently positive effects, with an average intensification of $\overline{\hat{\delta}}_{\mathrm{TIEE}}(0.99)\approx 2.18$~mm, exceeding and stabilising estimates from Zhang ($\approx 1.75$~mm) and Causal Hill ($\approx 1.78$~mm). 
It also reduces apparent spatial heterogeneity, suggesting that much of the variability in competing estimators reflects noise rather than genuine differences. 
Overall, once dynamical confounding is properly addressed, a robust thermodynamic signal of intensification emerges across the region.


\section{Conclusion} \label{sec:conclusion}
Estimating causal effects at extreme quantiles is central in applications such as climate science and econometrics, but standard methods fail in the tails. 
We introduce the Tail-Calibrated Inverse Estimating Equation (TIEE) framework, which integrates causal inference and extreme value theory via an integrated moment condition, yielding stable estimation without plug-in inversion. While implemented with a GPD, the framework is not model-specific and accommodates any uniformly consistent tail approximation.
We establish identification, consistency, and asymptotic normality, and simulations show substantial efficiency gains and robustness to propensity score misspecification. 
The method is compatible with flexible nuisance estimation, including machine learning approaches. 
An application to Austrian precipitation data reveals a clear intensification of extreme events after adjusting for atmospheric circulation.

Future work includes extensions to high-dimensional confounding, leveraging flexible machine learning estimators, and to multivariate or spatial extremes beyond the current univariate setting. 
The framework also extends naturally to continuous treatments by replacing the binary signal with a continuous treatment analogue based on the generalised propensity score or related doubly robust constructions~\citep{hirano2004propensity,kennedy2017non}, yielding an unbiased estimating function for the distribution function of the potential outcome at a given treatment level. 

More broadly, the TIEE framework addresses causal inference in sparse tails and has applications beyond climate science. 
In quantitative finance, it can be used to isolate the impact of policy interventions on systemic risk measures such as Value-at-Risk, and in epidemiology, to quantify intervention effects on peak infection rates, where average-based metrics can obscure tail behaviour. 
These extensions highlight the broad utility of tail-calibrated inference beyond environmental sciences.




\section{Acknowledgments}
Mengran Li gratefully acknowledges the support provided by the China Scholarship Council (CSC). The authors thank members of the Glasgow-Edinburgh Extremes Network (GLE$^2$N) for useful discussions related to this work.

\section{Data availability}
The data underlying this article are derived from publicly available sources. 
The ERA5 reanalysis data were obtained from the Copernicus Climate Change Service (C3S) Climate Data Store (https://cds.climate.copernicus.eu), including ERA5 hourly data on single levels (DOI:10.24381/cds.adbb2d47) and pressure levels (DOI:10.24381/cds.bd0915c6). 
The EEAR-Clim dataset is available from Bongiovanni et al. (2025). 
Derived data generated during the analysis are available from the corresponding author upon reasonable request.
The code to reproduce the simulation studies in Section~\ref{sec:simulation}, and the data and code to reproduce the data application in Section~\ref{sec:app} are freely available at \url{https://github.com/MengranLi-git/TIEE}.

\appendix

\section{Additional results and proofs from Section 4}\label{app:proofs}

\subsection{Proof of Lemma 1}

For each fixed $p\in(0,1)$, the mapping
$W \mapsto \1\{Q_{Y_d}(p)\le \theta\}$
is an indicator function indexed by the scalar threshold parameter $\theta$.
Such threshold indicator functions form a VC-subgraph class of finite VC dimension.
VC-subgraph classes have polynomial covering numbers and therefore satisfy the conditions of the Glivenko--Cantelli theorem~\citep[Sections~2.4 and 2.6]{vaart2023empirical}.
The integrated signal $S_{d}$ is obtained by integrating these indicator functions with respect to $p$ over a finite measure on $(0,1)$. 
Since the integrand is uniformly bounded and integration is a linear and bounded operator, the resulting class $\mathcal{F}$ inherits the Glivenko--Cantelli property~\citep[Section~2.4]{vaart2023empirical}.
Therefore,
\[
\sup_{\theta\in\Theta}
\left|
\frac{1}{n}\sum_{i=1}^n S_d(W_i,\theta)
-
\mathbb{E}\big[S_d(W,\theta)\big]
\right|
\;\xrightarrow{P}\; 0,
\]
and $\mathcal{F}$ is Glivenko--Cantelli.

\subsection{Proof of Theorem 3 (Asymptotic Normality)}

The proof of asymptotic normality for the TIEE estimator $\hat{\theta}_{d,\tau}$ follows the standard approach for semi-parametric Z-estimators. We assume that the true parameters $\theta_{d,\tau}$ and $\zeta$ are fixed, and we emphasise the dependence of the empirical functional $\Phi_{d,n}$ of the estimator of $\zeta$ by writing
\begin{equation*}
\Phi_{d,n}(\theta, \hat{\zeta}) = \frac{1}{n}\sum_{i=1}^{n}g_{\mathrm{TIEE},d}(W(i),\tau,\theta,\hat{\zeta}).
\end{equation*}
As we know, the estimator $\hat{\theta}_{d,\tau}$ is defined as the root of the empirical estimating equation $\Phi_{d,n}(\hat{\theta}_{d,\tau}, \hat{\zeta}) = 0.$
Applying the mean value theorem to $\Phi_{d,n}(\theta, \hat{\zeta})$ in $\theta$ around the true parameter $\theta_{d,\tau}$ yields
\begin{equation*}
0 = \Phi_{d,n}(\hat{\theta}_{d,\tau}, \hat{\zeta}) = \Phi_{d,n}(\theta_{d,\tau}, \hat{\zeta}) + (\hat{\theta}_{d,\tau}-\theta_{d,\tau}) \cdot \frac{\partial}{\partial \theta} \Phi_{d,n}(\tilde{\theta}, \hat{\zeta}),
\label{eq:taylor_expansion}
\end{equation*}
where $\tilde{\theta}$ is an intermediate value between $\hat{\theta}_{d,\tau}$ and $\theta_{d,\tau}$.
Rearranging terms, we get
\begin{equation}
\sqrt{n}(\hat{\theta}_{d,\tau}-\theta_{d,\tau}) = -\frac{\sqrt{n}\Phi_{d,n}(\theta_{d,\tau}, \hat{\zeta})}{\frac{\partial}{\partial \theta} \Phi_{d,n}(\tilde{\theta}, \hat{\zeta})}.
\label{eq:asymp_repr}
\end{equation}
By the consistency of $\hat{\theta}_{d,\tau}$ established in Theorem 2, we have $\tilde{\theta} \xrightarrow{P} \theta_{d,\tau}$ and by the assumptions of Theorem 2, $\hat{\zeta} \xrightarrow{P} \zeta$. Furthermore, the assumptions also ensure that the partial derivative $\frac{\partial}{\partial \theta} \Phi_d(\theta, \zeta)$ is continuous in a neighbourhood of $\theta_{d,\tau}$.
The uniform convergence of the empirical derivative to its population counterpart, 
\begin{equation*}
\frac{\partial}{\partial \theta} \Phi_{d,n}(\theta, \zeta) \xrightarrow{P} \frac{\partial}{\partial \theta}\Phi_d(\theta, \zeta),
\end{equation*}
combined with the consistency of the estimated parameters, ensures that
\begin{equation}
\frac{\partial}{\partial \theta} \Phi_{d,n}(\tilde{\theta}, \hat{\zeta}) \xrightarrow{P} \frac{\partial}{\partial \theta}\Phi_d(\theta_{d,\tau}, \zeta).
\label{eq:deriv_convergence}
\end{equation}
Since $\frac{\partial}{\partial \theta}\Phi_d(\theta_{d,\tau}) \neq 0$, the denominator in~\eqref{eq:asymp_repr} converges to a non-zero constant, ensuring the asymptotic normality result is well-defined.

Now we turn our attention to the numerator in~\eqref{eq:asymp_repr}.
Since this is a function of the estimated nuisance parameter $\hat{\zeta}$, we exploit the pathwise differentiability of the moment functional with respect to $\zeta$ to obtain the following asymptotic linear expansion around the true parameter $\zeta$:
\begin{equation}
\sqrt{n} \Phi_{d,n}(\theta_{d,\tau}, \hat{\zeta}) = \sqrt{n} \Phi_{d,n}(\theta_{d,\tau}, \zeta) + \sqrt{n}(\hat{\zeta}-\zeta)^{\top} \cdot \nabla_{\zeta} \Phi_d(\theta_{d,\tau}, \zeta) + o_{p}(1),
\label{eq:expansion_gamma}
\end{equation}
where $\nabla_{\zeta} \Phi_d(\theta_{d,\tau}, \zeta)$ is the gradient of the population functional with respect to $\zeta$, evaluated at the true parameters.
Since $\Phi_d(\theta_{d,\tau}, \zeta) = 0$ is the moment condition identifying the true quantile $\theta_{d,\tau}$, the moment condition is locally unbiased with respect to $\zeta$ at the root (this is a standard property for efficient semi-parametric estimators). 
This condition implies that the term involving the estimation error of $\hat{\zeta}$ vanishes asymptotically.
Therefore, the numerator in~\eqref{eq:expansion_gamma} simplifies to
\begin{equation*}
\sqrt{n} \Phi_{d,n}(\theta_{d,\tau}, \hat{\zeta}) = \sqrt{n} \Phi_{d,n}(\theta_{d,\tau}, \zeta) + o_{p}(1).
\label{eq:simplified_numerator}
\end{equation*}

Substituting the definition of $\Phi_{d,n}(\theta_{d,\tau}, \zeta)$ and using the identification condition $\E[{g_{\mathrm{TIEE},d}}(W(i), \tau, \theta_{d,\tau}, \zeta)] = \Phi_d(\theta_{d,\tau}, \zeta) = 0$, we obtain
\begin{equation*}
\sqrt{n} \Phi_{d,n}(\theta_{d,\tau}, \hat{\zeta}) = \frac{1}{\sqrt{n}} \sum_{i=1}^n {g_{\mathrm{TIEE},d}}(W_i, \tau, \theta_{d,\tau}, \zeta) + o_p(1),
\label{eq:sum_representation}
\end{equation*}
which, by the central limit theorem, converges to a mean-zero Gaussian distribution with finite variance $\sigma_{d,\tau}^{2} = \mathrm{Var}(g_{\mathrm{TIEE},d}(W_i, \tau, \theta_{d,\tau}, \zeta))$.
Combining (via Slutsky's theorem) this convergence result with that of~\eqref{eq:deriv_convergence}, we have
\begin{equation*}
\sqrt{n}(\hat{\theta}_{d,\tau}-\theta_{d,\tau}) = -\frac{\sqrt{n}\Phi_{d,n}(\theta_{d,\tau}, \hat{\zeta})}{\frac{\partial}{\partial \theta} \Phi_{d,n}(\tilde{\theta}, \hat{\zeta})} \xrightarrow{d}- \frac{\mathcal{N}(0, \sigma_{q}^{2})}{\frac{\partial}{\partial \theta}\Phi_d(\theta_{d,\tau}, \zeta)} \equiv \mathcal{N}\left(0, \frac{\sigma_{q}^{2}}{[\frac{\partial}{\partial \theta}\Phi_d(\theta_{d,\tau}, \zeta)]^{2}}\right).
\label{eq:asymp_normality}
\end{equation*}

\subsection{Variance of $\hat{\theta}_{d,\tau}$ when $\nabla_{\zeta}\Phi_d(\theta_{d,\tau},\zeta)\neq 0$}\label{app:variance}

\subsubsection{Derivation of $V_{d,\tau}$}
The standard derivation in Theorem~3 assumes that the estimation of the nuisance parameter $\zeta$ does not affect the asymptotic distribution of $\theta_{d,\tau}$ through the orthogonality condition $\nabla_{\zeta}\Phi_d(\theta_{d,\tau},\zeta)=0$. When this condition fails to hold, we must incorporate the estimation uncertainty of $\hat{\zeta}$ into the calculation of the asymptotic variance $V_{d,\tau}$. This section derives the general variance formula that accounts for the influence of nuisance parameter estimation.

Let $\zeta$ be a $k \times 1$ vector of nuisance parameters. The estimators $(\hat{\theta}_{d,\tau}, \hat{\zeta})$ are obtained by solving the empirical moment conditions
\begin{equation*}
\Phi_{d,n}(\hat{\theta}_{d,\tau},\hat{\zeta}) = \frac{1}{n}\sum_{i=1}^n g_{\mathrm{TIEE},d}(W_i,\hat{\theta}_{d,\tau},\hat{\zeta}) = 0,\quad\text{and}\quad
\Psi_{d,n}(\hat{\zeta}) = \frac{1}{n}\sum_{i=1}^n \psi(W_i,\hat{\zeta}) = 0, \label{eq:moment_gamma}
\end{equation*}
where $\psi(W_i,\hat{\zeta})$ is an estimating function defining the nuisance parameter estimator $\zeta$ through the moment condition $\E[\psi(W_i,\hat{\zeta})] = 0$.
Let the stacked moment function $m$ and its average be defined as
\begin{equation*}
m(W,\theta,\zeta) = \begin{pmatrix} g_{\mathrm{TIEE},d}(W,\theta,\zeta) \\ \psi(W,\zeta) \end{pmatrix}, \qquad \bar{m}_n(\theta,\zeta) = \frac{1}{n}\sum_{i=1}^n m(W_i,\theta,\zeta).
\label{eq:stacked_moment}
\end{equation*}
Under standard regularity conditions, the joint asymptotic distribution of $(\hat{\theta}_{d,\tau}, \hat{\zeta})$ is
\begin{equation*}
\sqrt{n} \begin{pmatrix} \hat{\theta}_{d,\tau}-\theta_{d,\tau} \\ \hat{\zeta}-\zeta \end{pmatrix} \xrightarrow{d} \mathcal{N}(0, V_{\text{full}}),
\label{eq:joint_asymp}
\end{equation*}
where $V_{\text{full}}$ is the sandwich covariance matrix given by $V_{\text{full}} = A^{-1} B (A^{-1})^{\top}$, where $A$ is the Jacobian matrix and $B$ is the covariance matrix.
Specifically, $A$ is the expected gradient of the moment function $m$ with respect to the parameters $(\theta, \zeta)$, i.e.,
\begin{equation*}
A = \E\left[ \frac{\partial m(W,\theta_{d,\tau}, \zeta)}{\partial (\theta, \zeta)} \right]
=
\begin{pmatrix}
\E[\partial_{\theta}g_{\mathrm{TIEE},d}] & \E[\partial_{\zeta}g_{\mathrm{TIEE},d}] \\
\E[\partial_{\theta}\psi] & \E[\partial_{\zeta}\psi]
\end{pmatrix}.
\label{eq:jacobian_full}
\end{equation*}
Under the natural assumption that the estimating equation for $\zeta$ does not depend on $\theta$, we have $\E[\partial_{\theta}\psi] = 0$. Therefore, the Jacobian matrix takes the block upper-triangular form
\begin{equation}
A =
\begin{pmatrix}
\Phi'(\theta_{d,\tau}) & A_{\zeta} \\  
0 & M
\end{pmatrix},
\label{eq:jacobian_block}
\end{equation}
where $\Phi'(\theta_{d,\tau}) = \E[\partial_{\theta}g_{\mathrm{TIEE},d}(W,\theta_{d,\tau},\zeta)]$ is a scalar representing the sensitivity of the moment condition to changes in $\theta$; $A_{\zeta} = \E[\partial_{\zeta}g_{\mathrm{TIEE},d}(W,\theta_{d,\tau},\zeta)] = \nabla_{\zeta}\Phi_d(\theta_{d,\tau},\zeta)$ is a $1 \times k$ row vector capturing the sensitivity to nuisance parameters; and $M = \E[\partial_{\zeta}\psi(W,\zeta)]$ is a $k \times k$ matrix governing the estimation of $\zeta$.

The matrix $B$ is the covariance matrix of the stacked moment functions evaluated at the true parameters, i.e.,
\begin{equation*}
B = \mathrm{Var}(m(W,\theta_{d,\tau},\zeta))
=
\begin{pmatrix}
\mathrm{Var}(g_{\mathrm{TIEE},d}) & \mathrm{Cov}(g_{\mathrm{TIEE},d}, \psi) \\
\mathrm{Cov}(\psi, g_{\mathrm{TIEE},d}) & \mathrm{Var}(\psi)
\end{pmatrix}
=
\begin{pmatrix}
\sigma_{d,\tau}^2 & \Sigma_{g,\psi} \\
\Sigma_{g,\psi}^{\top} & \Sigma_{\psi}
\end{pmatrix},
\label{eq:covariance_matrix}
\end{equation*}
where $\sigma_{d,\tau}^2 = \mathrm{Var}(g_{\mathrm{TIEE},d}(W,\theta_{d,\tau},\zeta))$ is the scalar variance of the moment function for $\theta_{d,\tau}$; $\Sigma_\psi = \mathrm{Var}(\psi(W,\zeta))$ is the $k \times k$ covariance matrix of the moment functions for $\zeta$; and $\Sigma_{g,\psi} = \mathrm{Cov}(g_{\mathrm{TIEE},d}(W,\theta_{d,\tau},\zeta),\psi(W,\zeta))$ is the $1 \times k$ row vector of covariances between the two sets of moment functions.

$V_{d,\tau}$ is defined by the $(1,1)$ block of $V_{\text{full}}$, i.e., $V_{d,\tau} = \mathrm{Var}(\sqrt{n}(\hat{\theta}_{d,\tau}-\theta_{d,\tau}))$.
For a block upper-triangular matrix of the form given in~\eqref{eq:jacobian_block}, the inverse is
\begin{equation}
A^{-1} = \begin{pmatrix} 
(\Phi')^{-1} & -(\Phi')^{-1} A_{\zeta} M^{-1} \\ 
0 & M^{-1} 
\end{pmatrix} 
= \begin{pmatrix} 
U_{11} & U_{12} \\ 
0 & U_{22} 
\end{pmatrix}.
\label{eq:inverse_jacobian}
\end{equation}
$V_{d,\tau}$ can then be computed as
\begin{equation}
V_{d,\tau} = U_{11} B_{11} U_{11}^{\top} + U_{12} B_{21} U_{11}^{\top} + U_{11} B_{12} U_{12}^{\top} + U_{12} B_{22} U_{12}^{\top},
\label{eq:sandwich_expansion}
\end{equation}
where $B_{11} = \sigma_{d,\tau}^2$, $B_{12} = \Sigma_{g,\psi}$, $B_{21} = \Sigma_{g,\psi}^{\top}$, and $B_{22} = \Sigma_\psi$.
Since $V_{d,\tau}$ is a scalar quantity, the second and third terms in~\eqref{eq:sandwich_expansion} are transposes of each other and therefore equal. Substituting the expressions for $U_{11}$ and $U_{12}$ from~\eqref{eq:inverse_jacobian}, we obtain the general variance formula
\begin{equation}
V_{d,\tau} = \frac{\sigma_{d,\tau}^2 + A_{\zeta} V_{\zeta} A_{\zeta}^{\top} - 2 A_{\zeta} C}{[\Phi'(\theta_{d,\tau})]^2},
\label{eq:Vq_general}
\end{equation}
where
$V_{\zeta} = M^{-1}\Sigma_{\psi}(M^{-1})^{\top}$ and $C = M^{-1}\Sigma_{g,\psi}^{\top}$.
Here, $V_{\zeta}$ represents the asymptotic variance of $\hat{\zeta}$, and $C$ captures the scaled covariance between the moment functions $\overline{g}$ and $\psi$.
The formula in~\eqref{eq:Vq_general} contains three components: $\sigma_{d,\tau}^2$, the intrinsic variability of the moment condition for $\theta_{d,\tau}$; $A_{\zeta} V_{\zeta} A_{\zeta}^{\top}$, the variance inflation due to the estimation of nuisance parameters; and $-2 A_{\zeta} C$, a correction term accounting for the correlation between the estimation procedures for $\theta_{d,\tau}$ and $\zeta$.

\subsubsection{Consistent plug-in estimation} 

A consistent plug-in estimator $\hat{V}_{d,\tau}$ can be constructed by replacing all population quantities in~\eqref{eq:Vq_general} with their empirical counterparts, i.e.,
\begin{equation*}
\hat{V}_{d,\tau} = \frac{\hat{\sigma}_{d,\tau}^2 + \hat{A}_{\zeta} \hat{V}_{\zeta} \hat{A}_{\zeta}^{\top} - 2 \hat{A}_{\zeta} \hat{C}}{(\hat{\Phi}')^2},
\label{eq:Vq_plugin}
\end{equation*}
where the empirical quantities are defined as:
\begin{align*}
\hat{\Phi}' &= \frac{1}{n}\sum_{i=1}^n \partial_{\theta} g_{\mathrm{TIEE},d}(W_i,\hat{\theta}_{d,\tau},\hat{\zeta}),\\
\hat{A}_{\zeta} &= \frac{1}{n}\sum_{i=1}^n \partial_{\zeta} g_{\mathrm{TIEE},d}(W_i,\hat{\theta}_{d,\tau},\hat{\zeta}), \\
\hat{M} &= \frac{1}{n}\sum_{i=1}^n \partial_{\zeta}\psi(W_i;\hat{\zeta}),\\
\hat{V}_{\zeta} &= \hat{M}^{-1}\hat{\Sigma}_{\psi}(\hat{M}^{-1})^{\top},\\
\hat{C} &= \hat{M}^{-1}\hat{\Sigma}_{g,\psi}^{\top}.
\end{align*}

The empirical second moments $\hat{\sigma}_{d,\tau}^2$, $\hat{\Sigma}_{\psi}$, and $\hat{\Sigma}_{g,\psi}$ are computed from the centred moment functions as
\begin{align*}
\hat{\sigma}_{d,\tau}^2 &= \frac{1}{n}\sum_{i=1}^n [g_{\mathrm{TIEE},d}(W_i,\hat{\theta}_{d,\tau},\hat{\zeta})]^2,\\
\hat{\Sigma}_{\psi} &= \frac{1}{n}\sum_{i=1}^n \psi(W_i,\hat{\zeta})\psi(W_i,\hat{\zeta})^{\top},\\
\hat{\Sigma}_{g,\psi} &= \frac{1}{n}\sum_{i=1}^n g_{\mathrm{TIEE},d}(W_i,\hat{\theta}_{d,\tau},\hat{\zeta})\psi(W_i,\hat{\zeta})^{\top}.
\end{align*}

Under standard regularity conditions, $\hat{V}_{d,\tau} \xrightarrow{P} V_{d,\tau}$, ensuring consistent variance estimation for constructing confidence intervals and hypothesis tests.

\subsubsection{Special Cases}

The general variance formula in~\eqref{eq:Vq_general} encompasses two important special cases that arise in practice.

\noindent\textbf{Case 1: Orthogonality.} If the orthogonality condition $\nabla_{\zeta}\Phi_d(\theta_{d,\tau},\zeta)=0$ holds, as assumed in Theorem~3, then $A_{\zeta}=0$. Consequently, both the variance inflation term $A_{\zeta}V_{\zeta}A_{\zeta}^{\top}$ and the cross-term $2A_{\zeta}C$ vanish, and the variance simplifies to
\begin{equation*}
V_{d,\tau} = \frac{\sigma_{d,\tau}^2}{[\Phi'(\theta_{d,\tau})]^2},
\label{eq:Vq_orthogonal}
\end{equation*}
recovering the simpler result stated in Theorem~3. This is the most efficient case, as the nuisance parameter estimation introduces no additional variability.

\noindent\textbf{Case 2: Uncorrelatedness.} If $\Sigma_{g,\psi}=0$ then $C=0$ and the cross-term $2A_\zeta C$ disappears. This case can arise, for example, through sample-splitting or cross-fitting procedures that empirically decouple the estimation of $g_{\mathrm{TIEE},d}$ and $\psi$.
In this case, the variance becomes
\begin{equation*}
V_{d,\tau} = \frac{\sigma_{d,\tau}^2 + A_{\zeta} V_{\zeta} A_{\zeta}^{\top}}{[\Phi'(\theta_{d,\tau})]^2}.
\label{eq:Vq_uncorrelated}
\end{equation*}
This case exhibits only the variance inflation due to nuisance parameter estimation, without the potentially beneficial correction from correlated estimation errors.

\subsection{Asymptotic derivations for the extreme quantile treatment effect $\delta(\tau)$}
In this section, we provide detailed derivations supporting the asymptotic distribution of $\delta(\tau) = \theta_{1,\tau} - \theta_{0,\tau}$ provided in Section~\ref{sec:eqte_inference}.
We start by deriving the joint asymptotic distribution of $(\theta_{1,\tau},\theta_{0,\tau})^\top$.
A first order Taylor expansion of the empirical estimating equation $\Phi_{d,n}(\widehat{\theta}_{d,\tau}) = 0$ around $\theta_{d,\tau}$
gives
\[
0 = \Phi_{d,n}(\widehat{\theta}_{d,\tau}) \approx \Phi_{d,n}(\theta_{d,\tau}) + \Phi'_{d,n}(\theta_{d,\tau})(\hat{\theta}_{d,\tau} - \theta_{d,\tau}), \quad d = 0, 1,
\]
which means that $\sqrt{n}(\hat{\theta}_d - \theta_{d,\tau}) \approx -\frac{\sqrt{n} \Phi_{d,n}(\theta_{d,\tau})}{\Phi'_{d,n}(\theta_{d,\tau})}.$
By consistency and uniform convergence, $\Phi'_{d,n}(\theta_{d,\tau})\xrightarrow{P}\Phi'_{d}(\theta_{d,\tau})$, so asymptotically,
\begin{equation*}\label{eq:expansion_joint2}
\sqrt{n}(\hat{\theta}_d - \theta_{d,\tau}) = -\frac{\sqrt{n} \Phi_{d,n}(\theta_{d,\tau})}{\Phi'_{d}(\theta_{d,\tau})} + o_p(1).
\end{equation*}
Now, stacking the two treatment groups together, we get
\[
\sqrt{n}
\begin{pmatrix}
\hat{\theta}_{1,\tau} - \theta_{1,\tau} \\
\hat{\theta}_{0,\tau} - \theta_{0,\tau}
\end{pmatrix}
=
-\,D^{-1}
\sqrt{n}
\begin{pmatrix}
\Phi_{n,1}(\theta_{1,\tau}) \\
\Phi_{n,0}(\theta_{0,\tau})
\end{pmatrix}
+ o_p(1),
\qquad
D = \mathrm{diag}\!\big(\Phi'_1(\theta_{1,\tau}), \Phi'_0(\theta_{0,\tau})\big).
\]

By the multivariate central limit theorem,
\[
\sqrt{n}
\begin{pmatrix}
\Phi_{n,1}(\theta_{1,\tau}) \\
\Phi_{n,0}(\theta_{0,\tau})
\end{pmatrix}
\xrightarrow{d}
\mathcal{N}(0, \Sigma_\Phi),
\]
where
\begin{equation}\label{eq:Sigma_Phi}
\Sigma_\Phi = \begin{pmatrix}
\Sigma_{11} & \Sigma_{10} \\
\Sigma_{10} & \Sigma_{00}
\end{pmatrix}
= \begin{pmatrix}
\text{Var}(g_{1}(W,\theta_{1,\tau})) & \text{Cov}(g_{1}(W,\theta_{1,\tau}),g_{0}(W,\theta_{0,\tau})) \\
\text{Cov}(g_{1}(W,\theta_{1,\tau}),g_{0}(W,\theta_{0,\tau})) & \text{Var}(g_{0}(W,\theta_{0,\tau}))
\end{pmatrix}.
\end{equation}
In~\eqref{eq:Sigma_Phi}, we have reduced notation clutter by denoting $g_{d}(W,\theta_{d,\tau})) = g_{\mathrm{TIEE},d}(W,\theta,\tau,\zeta)$, so $\text{Var}(g_{d}(W,\theta_{d,\tau}))$ is the variance of the integrated estimating function for treatment $d=0,1$, which can be estimated using the sample variance:
\begin{equation}\label{eq:Sigma_dd_estimator}
\hat{\Sigma}_{dd} = \frac{1}{n} \sum_{i=1}^n g_{d}(W(i), \hat{\theta}_{d,\tau})^2 - \left(\frac{1}{n} \sum_{i=1}^n g_{d}(W(i), \hat{\theta}_{d,\tau}) \right)^2 \approx \frac{1}{n} \sum_{i=1}^n g_{d}(W(i), \hat{\theta}_{d,\tau})^2.
\end{equation}
The approximation in~\eqref{eq:Sigma_dd_estimator} comes from the fact that $\widehat{\theta}_{d,\tau}$ solves $\Phi_{d,n}(\widehat{\theta}_{d,\tau}) = 0$, therefore the sample mean is close to 0.
The covariance term $\Sigma_{10}$ captures the dependence between the estimating equations for the two treatment groups. When treatment and control samples are independent (e.g., in randomised experiments with non-overlapping units), $\Sigma_{10} = 0$.
However, in observational studies with IPW adjustment, $\Sigma_{10} \neq 0$ in general because both estimators use the same propensity score model and the same covariate distribution. 
In that case, the covariance can be estimated by
\begin{equation*}\label{eq:Sigma_10_estimator}
\hat{\Sigma}_{10} = \frac{1}{n} \sum_{i=1}^n g_{1}(W(i), \hat{\theta}_{1,\tau}) g_{0}(W(i), \hat{\theta}_{0,\tau})  - \left(\frac{1}{n} \sum_{i=1}^n g_{1}(W(i), \hat{\theta}_{1,\tau}) \right)\left(\frac{1}{n} \sum_{i=1}^n g_{0}(W(i), \hat{\theta}_{0,\tau}) \right).
\end{equation*}

Therefore, the joint re-scaled asymptotic distribution of $(\widehat{\theta}_{1,\tau},\widehat{\theta}_{0,\tau})^\top$ is given by
\[
\sqrt{n}
\begin{pmatrix}
\hat{\theta}_{1,\tau} - \theta_{1,\tau} \\
\hat{\theta}_{0,\tau} - \theta_{0,\tau}
\end{pmatrix}
\xrightarrow{d}
\mathcal{N}\!\left(0, D^{-1} \Sigma_\Phi D^{-1}\right).
\]
The asymptotic distribution of $\hat{\delta}(\tau)$ is therefore $\sqrt{n}\big(\hat{\delta}(\tau) - \delta(\tau)\big)
\xrightarrow{d}
\mathcal{N}\!\left(
0,\; \sigma_\delta^2\right),$
where the asymptotic variance $\sigma_\delta^2 = (1,-1)\, D^{-1}\Sigma_\Phi D^{-1}\,(1,-1)^\top$ simplifies to
\begin{equation}\label{eq:eqte_variance}
    \sigma_\delta^2
=
\frac{\Sigma_{11}}{[\Phi'_1(\theta_{1,\tau})]^2}
+
\frac{\Sigma_{00}}{[\Phi'_0(\theta_{0,\tau})]^2}
-
\frac{2\Sigma_{10}}{\Phi'_1(\theta_{1,\tau})\Phi'_0(\theta_{0,\tau})}.
\end{equation}
A consistent estimator $\hat{\sigma}_\delta^2$ of~\eqref{eq:eqte_variance} is obtained by substituting the terms on the right-hand side by sample analogues.
The terms in the denominators can either be estimated via numerical differentiation or kernel density estimation of $f_{Y_d}(\hat{\theta}_{d,\tau})$ (since $\Phi'_d(\theta_{d,\tau}) = f_{Y_d}(\hat{\theta}_{d,\tau})$).
An asymptotically valid $(1 - \alpha)$ confidence interval for $\widehat{\delta}(\tau)$ is therefore given by $\widehat{\delta}(\tau) \pm z_{1-\alpha/2} {\hat{\sigma}_\delta}/{\sqrt{n}},$
where $z_{1-\alpha/2}$ denotes the $(1 - \alpha/2)$-quantile of the standard normal distribution.

\section{Additional simulation results and studies}\label{app:sensitive}
\subsection{Performance under heavy- and light-tailed scenarios with $n=5000$}\label{app:n5000}
The results for sample size $n=5000$ presented in Table~\ref{tab:heavy_light_tail_5000} exhibit the same patterns as those reported for $n=1000$. In particular, the TIEE-IPW estimator continues to outperform competing methods across all tail regimes, with the performance gains most pronounced in the extreme tail. As expected, increasing the sample size leads to overall reductions in bias and MSE for all estimators, but does not alter the relative ranking of the methods. These findings further support the robustness and scalability of the proposed approach.

\begin{table}[htbp]
\centering
\scriptsize
\caption{Estimation performance (Bias and MSE) under heavy- and light-tailed scenarios for sample size $n=5000$.}
\label{tab:heavy_light_tail_5000}
\resizebox{\textwidth}{!}{%
\begin{tabular}{llrrrrrr}
\toprule
\multicolumn{8}{c}{Heavy-tailed scenarios} \\
\midrule
 & & \multicolumn{2}{c}{Scenario $M_1^{\text{(H)}}$} & \multicolumn{2}{c}{Scenario $M_2^{\text{(H)}}$} & \multicolumn{2}{c}{Scenario $M_3^{\text{(H)}}$} \\
\cmidrule(lr){3-4} \cmidrule(lr){5-6} \cmidrule(lr){7-8}
$1-\tau_n$ & Method & Bias & MSE & Bias & MSE & Bias & MSE \\
\midrule
\multirow{4}{*}{$5/(n \log n)$}
 & Causal Hill  & 40.165 & 3652.733   & 22.858  & 1537.911    & 35.070  & 3914.050    \\
 & Pickands     & 70.235 & 29395.493  & 130.364 & 202295.849  & 153.724 & 354024.912  \\
 & Zhang-Firpo  & 33.818 & 4655.698   & 42.605  & 14786.807   & 39.325  & 9158.521    \\
 & TIEE-IPW     & 22.842 & \textbf{876.524} & 22.348 & \textbf{996.323} & 21.852 & \textbf{878.432} \\
\cmidrule{1-8}
\multirow{4}{*}{$1/n$}
 & Causal Hill  & 31.616 & 2279.542   & 18.035 & 942.085    & 26.499  & 2181.550    \\
 & Pickands     & 59.719 & 15831.284  & 95.367 & 86841.475  & 107.704 & 139943.738  \\
 & Zhang-Firpo  & 31.558 & 4507.900   & 40.327 & 14597.895  & 35.425  & 8866.979    \\
 & TIEE-IPW     & 19.310 & \textbf{695.274} & 18.341 & \textbf{761.169} & 17.668 & \textbf{632.377} \\
\cmidrule{1-8}
\multirow{4}{*}{$5/n$}
 & Causal Hill  & 8.553  & 182.638  & 5.164 & \textbf{72.501} & 6.128 & 107.085          \\
 & Pickands     & 26.375 & 1308.698 & 20.979 & 1807.624       & 19.685 & 1832.421        \\
 & Zhang-Firpo  & 9.558  & 272.607  & 7.824 & 208.126        & 8.530 & 447.514          \\
 & TIEE-IPW     & 7.061  & \textbf{131.785} & 5.947 & 114.472 & 5.163 & \textbf{75.775} \\
\midrule
\multicolumn{8}{c}{Light-tailed scenarios} \\
\midrule
 & & \multicolumn{2}{c}{Scenario $M_1^{\text{(L)}}$} & \multicolumn{2}{c}{Scenario $M_2^{\text{(L)}}$} & \multicolumn{2}{c}{Scenario $M_3^{\text{(L)}}$} \\
\cmidrule(lr){3-4} \cmidrule(lr){5-6} \cmidrule(lr){7-8}
$1-\tau_n$ & Method & Bias & MSE & Bias & MSE & Bias & MSE \\
\midrule
\multirow{4}{*}{$5/(n \log n)$}
 & Causal Hill  & 17.536  & 339.969 & 6.735  & 71.926  & 0.797  & 0.826  \\
 & Pickands     & -20.353 & 492.972 & -1.767 & 492.544 & -2.236 & 10.139 \\
 & Zhang-Firpo  & -0.770  & 11.042  & 0.196  & 13.396  & -0.386 & 0.322  \\
 & TIEE-IPW     & -1.688 & \textbf{7.863} & -0.683 & \textbf{5.600} & -0.173 & \textbf{0.138} \\
\cmidrule{1-8}
\multirow{4}{*}{$1/n$}
 & Causal Hill  & 13.458  & 201.658 & 4.997  & 40.885  & 0.628  & 0.535 \\
 & Pickands     & -19.186 & 430.554 & -2.139 & 293.848 & -2.118 & 8.627 \\
 & Zhang-Firpo  & 0.316   & 10.549  & 0.847  & 14.074  & -0.262 & 0.242 \\
 & TIEE-IPW     & -1.456 & \textbf{5.861} & -0.486 & \textbf{3.930} & -0.150 & \textbf{0.104} \\
\cmidrule{1-8}
\multirow{4}{*}{$5/n$}
 & Causal Hill  & 4.903   & 28.551  & 1.604  & 5.565  & 0.250  & 0.113 \\
 & Pickands     & -15.105 & 259.551 & -2.340 & 64.948 & -1.685 & 4.994 \\
 & Zhang-Firpo  & -0.169  & 2.942   & -0.116 & 2.348  & -0.044 & 0.070 \\
 & TIEE-IPW     & -0.876 & \textbf{2.081} & -0.076 & \textbf{1.095} & -0.105 & \textbf{0.039} \\
\bottomrule
\end{tabular}%
}
\end{table}

\subsection{Sensitivity to the threshold $u$}\label{app:sensitive_u}

The TIEE estimator's performance depends on the appropriate selection of the intermediate quantile level $p_u$ associated with the threshold $u$. To assess this sensitivity, we fix the sample size at $n=1000$ and consider the target quantile levels $\tau_n$ with $(1-\tau_n)\in \{5/n, 1/n, 5/(n \log n)\}$. We consider $p_u\in[0.85,0.95]$ with increments of 1\% within each treatment group (i.e., $u$ varies over the 85th and 95th percentiles of the control distribution).

Figure~\ref{fig:u_plot} reports the mean squared error of the estimated EQTE ${\delta}(\tau_n)$ (see its definition in Section~\ref{sec:QTE_summary}) as a function of $p_u$.
for different tail probability regimes. As expected, the absolute MSE increases as the target tail probability decreases, reflecting the increasing extremeness of the target quantile and the reduced effective sample size available for tail calibration. 
However, the primary goal of this analysis is to assess the sensitivity of the estimator to the choice of threshold.
From this perspective, we can see that across the range $p_u\in[0.85,0.95]$ the MSE varies smoothly without exhibiting sharp increases or pronounced minima, indicating that the estimator does not depend critically on fine tuning of the threshold. While very high thresholds increase variability due to the limited number of exceedances and very low thresholds may introduce bias from tail-model misspecification, the estimator remains reasonably stable across a broad intermediate range of thresholds.

\begin{figure}[htbp]
    \centering
    \includegraphics[width=0.8\textwidth]{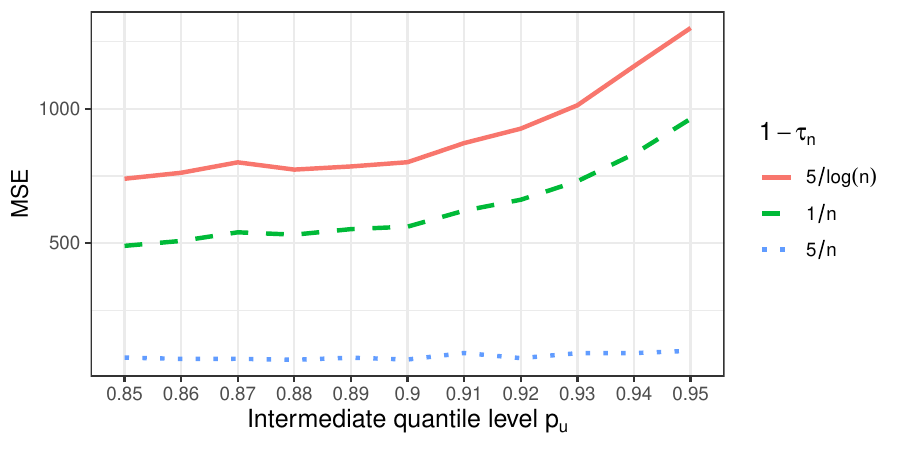} 
    \caption{Mean squared error of the TIEE-IPW estimator of the EQTE $\delta(\tau_n)$ as a function of the intermediate quantile level $p_u$ for $(1-\tau_n)\in\{5/\log(n), 1/n, 5/n\}$.}
    \label{fig:u_plot}
\end{figure}

\subsection{Robustness to grid size $K$\label{sec:ksize}}

Under the same experimental setting as Section~\ref{app:sensitive_u}, we investigate robustness with respect to the grid size used in the numerical integration detailed in Section~\ref{sec:tiee_numerical}. 
We vary the grid resolution over $\{100,200,\dots,2000\}$ while fixing the threshold $u$ at $n^{0.65}$ and keeping all other components unchanged.

Figure \ref{fig:k_plot} examines the robustness of the TIEE-IPW estimator of $\delta(\tau_n)$ to the grid size $K$ under three tail probability regimes.
Across the range of grid resolutions considered, the MSE remains stable, with no systematic deterioration as $K$ increases. 
This indicates that the performance of TIEE is not driven by fine-tuning of the grid resolution and that relatively coarse grids already provide sufficient accuracy. 
The result suggests that the integrated estimating equation is numerically well behaved and that the proposed implementation is robust to the discretisation choices required for practical computation.
\begin{figure}[H]
    \centering
    \includegraphics[width=0.9\textwidth]{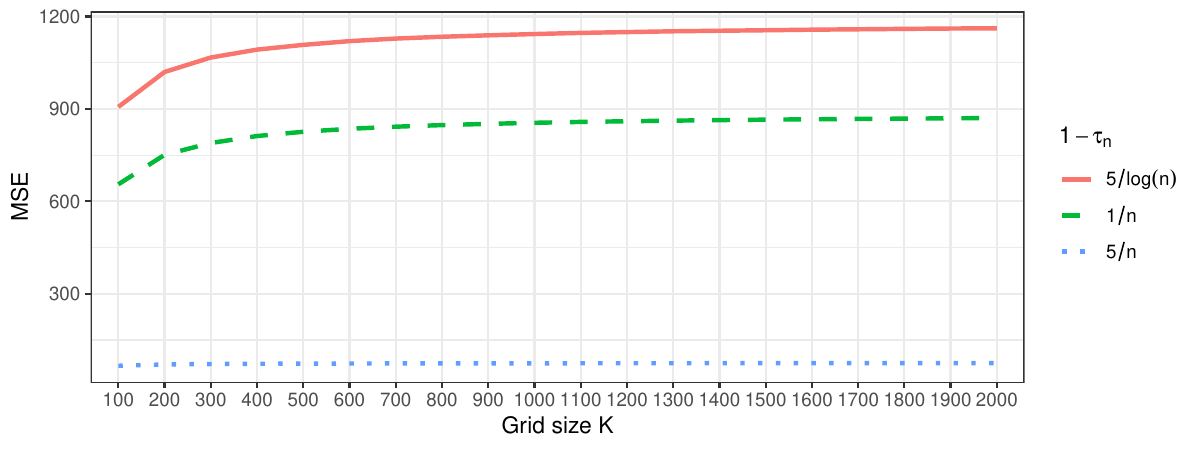} 
    \caption{Mean squared error (MSE) of the TIEE-IPW estimator of the EQTE $\delta(\tau_n)$ as a function of the grid size $K$ for $(1-\tau_n)=5/\log(n), 1/n, 5/n$.} 
    \label{fig:k_plot}
\end{figure}

\baselineskip 14pt
\bibliographystyle{CUP}
\bibliography{cite}

\end{document}